\documentclass[11pt, onecolumn, draftclsnofoot]{IEEEtran}
\usepackage{amsmath,amssymb}
\usepackage[dvips]{graphicx}
\usepackage{amsfonts}
\usepackage[mathscr]{eucal}
\usepackage{latexsym}
\usepackage{amsthm}
\usepackage{exscale}
\usepackage[mathscr]{eucal}
\usepackage{bm}
\usepackage[dvipsnames]{color}
\usepackage{cases}
\usepackage{epsfig}
\usepackage[center,small]{caption}
\usepackage{algorithm}
\usepackage{algorithmic}
\usepackage[verbose,nospace,sort]{cite}
\usepackage{tabularx}
\usepackage{multirow}
\usepackage{balance}
\usepackage{url}
\usepackage{graphicx}
\usepackage{subfig}
\usepackage[export]{adjustbox}

\scrollmode

\newtheorem{lemma}{Lemma}

\newtheorem{example}{Example}

\IEEEoverridecommandlockouts

\begin{document}
\title{Minimum-Latency FEC Design with Delayed Feedback: Mathematical Modeling and Efficient Algorithms }
\author{Xiaoli~Xu, Yong Zeng, \emph{Member, IEEE,}  Yonghui~Li, \emph{Fellow, IEEE}, and Branka~Vucetic, \emph{Fellow, IEEE}
\thanks{X. Xu and Y. Zeng (Corresponding author) are with the National Mobile Communications Research Laboratory, School of Information Science and Engineering, Southeast University, Nanjing 210096, China (email: {xiaolixu, yong\_zeng}@seu.edu.cn). }
\thanks{ Y. Li and B. Vucetic are with the School of Electrical and Information Engineering, The University of Sydney,  Sydney,  Australia (email: \{yonghui.li, branka.vucetic\}@sydney.edu.au).}
}

\maketitle
\vspace{-0.5in}
\begin{abstract}
In this paper, we consider the packet-level forward error correction (FEC) code design, without feedback or with delayed feedback, for achieving the minimum \emph{end-to-end latency}, i.e., the latency between the time that packet is generated at the source and its \emph{in-order delivery} to the application layer of the destination. We first show that the minimum-latency FEC design problem can be modeled as a partially observable Markov decision process (POMDP), and hence the optimal code construction can be obtained by solving the corresponding POMDP. However, solving the POMDP optimally is in general difficult unless the size is very small. To this end, we propose an efficient heuristic algorithm, namely the majority vote policy, for obtaining a high quality approximate solution. We also derive the tight lower and upper bounds of the optimal state values of this POMDP, based on which a more sophisticated $D$-step search algorithm can be implemented for obtaining near-optimal solutions.  The simulation results show that the proposed code designs via solving the POMDP, either with the majority vote policy or the $D$-step search algorithm, strictly outperform the existing schemes, in both cases, without or with only delayed feedback.
\end{abstract}


\section{Introduction}\label{sec:intro}
In many practical applications,  such as real-time video streaming, intelligent transportation and Tactile Internet, the \emph{end-to-end latency} \cite{3GPPLatency}, i.e., the latency between the time that a packet is generated at the source and its \emph{in-order delivery} to the application layer of the destination, is one of the most important performance metrics of a communication system. Establishing reliable connections with low end-to-end latency is one of the key 5G communications services defined by the International Telecommunication Union
(ITU), which is usually referred to as ultra-reliable and low-latency communication (URLLC) \cite{3GPPGeneral}. Most of the existing research on 5G URLLC has focused on physical layer or link layer technologies for improving the latency performance, such as short length channel codes \cite{Mahyar2018} and resource allocation \cite{she2017radio}. However, it was recently revealed in \cite{ETSI} that the transportation layer also needs to be re-designed to meet the critical latency requirement. Many of the existing reliable transport protocols, such as TCP, use the selective repeat automatic-repeat query (ARQ) to enhance the communication reliability. However, the latency performance of ARQ is largely dependent on the feedback delay \cite{Xia2003}. Specifically, when the feedback delay is negligible as compared with the time duration for sending one packet, ARQ can achieve both the maximal throughput and minimum end-to-end latency. Unfortunately, in many applications, such as the satellite communication and underwater acoustic communication, the feedback delay is tremendous due to the long signal propagation delays. Moreover, in some other applications, e.g., the bandwidth or time sensitive applications, we may have no feedback at all due to either unaffordable overhead for establishing the feedback link or unacceptable delay introduced by the feedback link. In the scenarios with delayed feedback or without feedback, the packet-level forward error correction (FEC) code is generally adopted in the transportation layer for enhancing the reliability and latency performance \cite{Sundararajan2009,Cloud2015,LowDelayFEC,Xu2016, Xu2018, Gabriel2018_2,Garrido2018,Sundararajan2017,Malak2019}.

\subsection{End-to-End Latency}
As shown in Fig.~\ref{F:blockDiagram}, in a coded communication system, the \emph{end-to-end latency} of a packet\cite{Cloud2015,LowDelayFEC,Gabriel2018_2,Garrido2018,Sundararajan2017}, denoted as $D_{\mathrm{e2e}}$, consists of three parts:
\begin{itemize}
\item{\emph{Queueing delay $D_q$}: the time since the packet being generated till its first attempt of being transmitted by the transmitter.}
\item{\emph{Decoding delay $D_c$}: the time since the packet being transmitted for the first time till that it is successfully decoded by the receiver.  }
\item{\emph{In-order delivery delay $D_d$}: the time that the packet spends in the buffer at the receiver awaiting in-order delivery to the application layer.}
\end{itemize}
\begin{figure}[htb]
\centering
\includegraphics[scale=0.5]{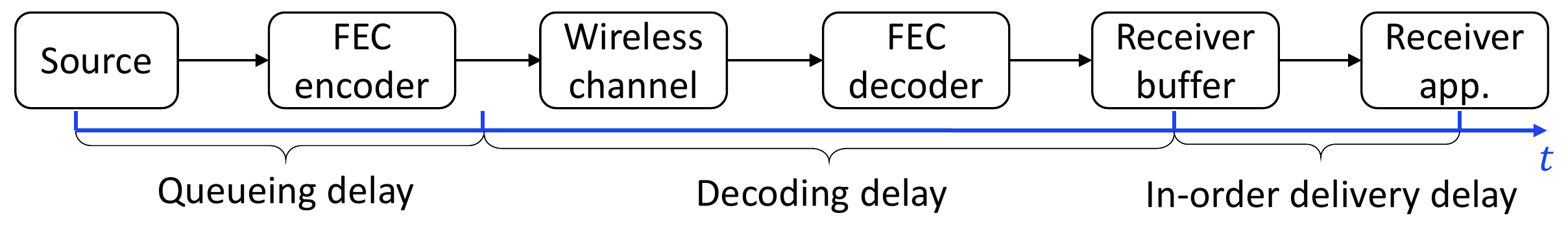}
\caption{A block diagram of the communication systems with coded transportation layer. }
\label{F:blockDiagram}
\end{figure}
A simple example for illustrating the definition of the various delays is presented in Fig.~\ref{F:delay}. Let us assume that there are 9 packets generated at $t=0$ by the source, denoted by $x_1,x_2,...,x_9$. The time for sending one packet is referred to as one time slot and the propagation delay is neglected. We consider a simple scheme to combat the packet erasures, for which a coded packet over all the preceding packets are periodically sent after the transmission of three information packets \cite{LowDelayFEC}. The coded packets are assumed to be generated via random linear network coding (RLNC) from a sufficiently large field, and hence the packets can be successfully decoded  when the number of innovative packets\footnote{A packet is innovative if it is independent of all other packets in the buffer.} collected by the receiver is equal to the number of information packets coded together. As shown in Fig.~\ref{F:delay}(a), the first three packets, which are assumed to be successfully received, are delivered to the application layer immediately after decoding. On the other hand, since packets $x_4$ and $x_5$ are erased, though packet $x_6$ is decoded at the end of time slot 7, it can only be delivered to the application layer after all the preceding packets are successfully decoded. This will happen by the end of slot 12 when $x_4$ and $x_5$ are decoded.  Hence, the in-order delivery delay for $x_6$ is 5 time slots. The delays experienced by each packet is summarized in Fig.~\ref{F:delay}(b).
\begin{figure}[h]
\centering
\subfloat[][]{%
\adjustbox{valign=b}{
\includegraphics[width=0.5\textwidth]{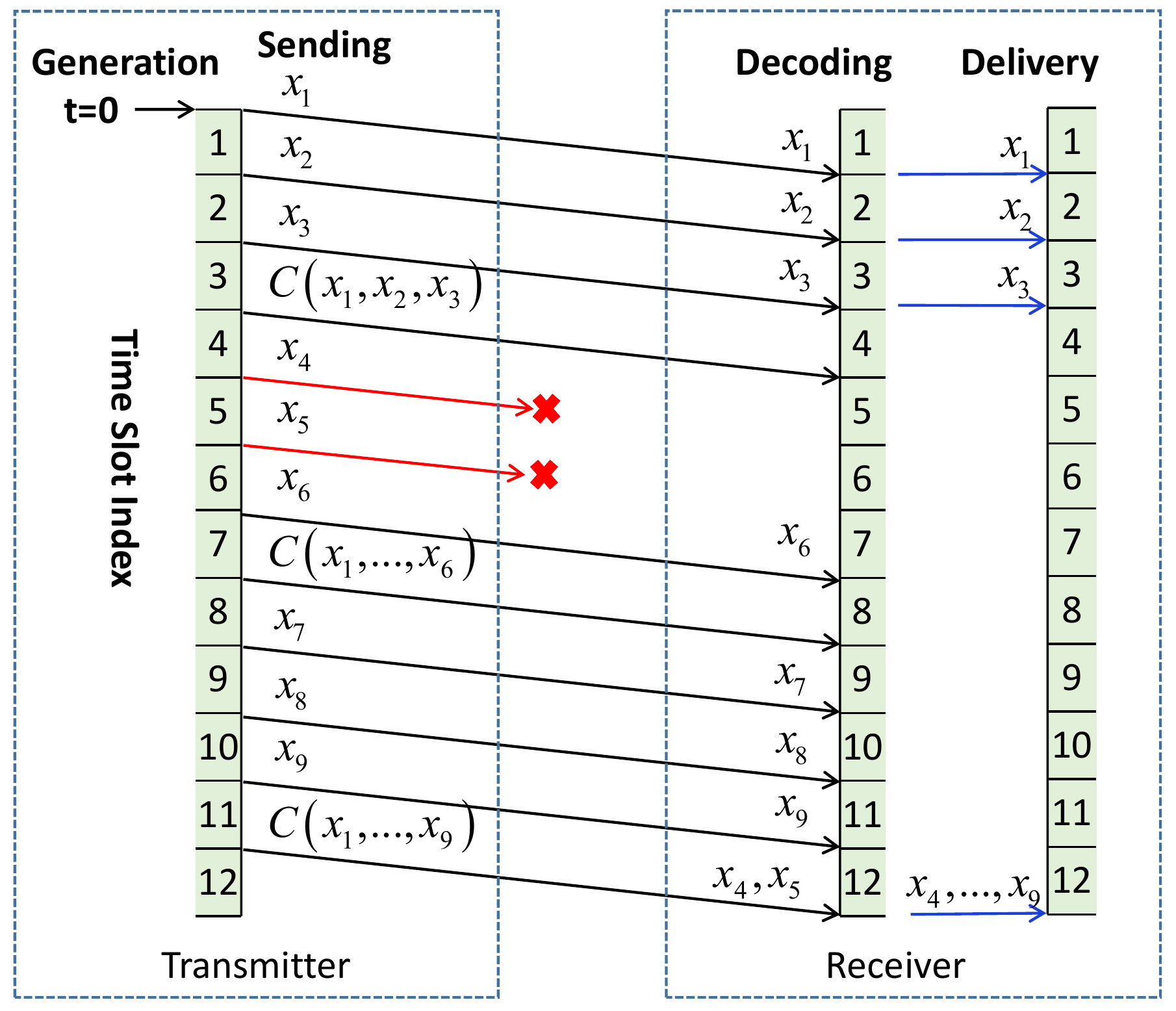}}}
\hspace{0.03cm}
\subfloat[][]{%
\adjustbox{valign=b}{
\begin{tabular}{|c|c|c|c|c|}
\hline
Packet & $D_q$ & $D_c$ & $D_d$ & $D_{\mathrm{e2e}}$ \\
\hline
$x_1$ & 0 & 1  & 0 & 1\\
\hline
$x_2$ & 1 & 1  & 0 & 2\\
\hline
$x_3$ & 2 & 1  & 0 & 3\\
\hline
$x_4$ & 4 & 8  & 0 & 12\\
\hline
$x_5$ & 5 & 7  & 0 & 12\\
\hline
$x_6$ & 6 & 1  & 5 & 12\\
\hline
$x_7$ & 8 & 1  & 3 & 12\\
\hline
$x_8$ & 9 & 1  & 2 & 12\\
\hline
$x_9$ & 10 & 1  & 0 & 12\\
\hline
\end{tabular}}}
\caption{An illustration of the end-to-end latency and its three delay components. (a) A transmission example where black lines denote successful transmission, red lines denote packet lost and blue lines denotes deliver to the application layer; (b) A summary of delays experienced by each packet. }
\label{F:delay}
\end{figure}

Consider a group of $N$ packets, where the end-to-end delay experienced by the $i$th packet is denoted as $D_{\mathrm{e2e}}(i)$, where $i=1,...,N$. In this paper, we focus on the FEC design that minimizes the average end-to-end latency experienced by all the packets, which is defined as
\begin{align}
 \bar{D}_{\mathrm{e2e}}\triangleq\frac{1}{N}\sum_{i=1}^{N}D_{\mathrm{e2e}}(i)=\frac{1}{N}\sum_{i=1}^{N}\left(D_{q}(i)+D_c(i)+D_d(i)\right).\label{eq:expectedDelay}
  \end{align}
 If instantaneous feedback from the receiver to the transmitter is available, achieving the minimum end-to-end delay is straightforward, i.e., by using ARQ \cite{Cloud2015}. However, for those practical  scenarios in the absence of any feedback or with only delayed feedback, achieving the minimum end-to-end delay is much more challenging, which thus motivates the current work. Note that for end-to-end latency minimization where in-order delivery is essential, decoding of a packet when any of its preceding packets remains undecodable will not improve the latency performance. Therefore, without loss of optimality, we may assume that the coded packets are constructed by combining all those unacknowledged information packets sent in previous time slots. As a result, the code design is simplified to deciding when a coded packet should be sent. Intuitively, the more coded packets sent at the early stage, the higher queueing delay it will cause to  the subsequent packets. On the other hand, if insufficient coded packets are sent,  there is a high probability that some packets remain undecodable for a long time, which increases the in-order delivery delay for all subsequent packets. Hence, determining how frequent and when a coded packet should be sent is highly non-trivial, which requires sequential decision making adaptive  to the transmission process. Specifically, in the absence of any feedback, it is desirable to send more information packets at the beginning of the transmission process so as to reduce the queueing delay, whereas send more coded packets at the later stage to reduce the in-order delivery delay. On the other hand, if feedback information is available, we can obtain more accurate estimation about the number of waiting packets at the receiver buffer, which helps to make the coding decision at the transmitter side.

\subsection{Related Work}\label{sec:related}
Adopting FEC in transport layer of communication systems for enhancing the end-to-end latency performance has received significant research attention. Specifically,  in \cite{Sundararajan2009}, the authors proposed a novel architecture by including a new network coding (NC) layer between the TCP and IP layers. For each arrival packet from TCP, a fixed number of network coded packets are generated by using the recent packets stored in the buffer, and the buffer is updated based on the feedback from the receiver. Later in \cite{Cloud2015}, the authors investigated the integration of FEC with ARQ, where the coding rate of the FEC was adaptively chosen according to the delayed feedback. The achievable end-to-end latency was derived as a function of the coding rate and feedback delay. Recently, a low-latency sliding-window network code was proposed in \cite{LowDelayFEC}, where a coded packet is inserted after every $L-1$ information packets. The average end-to-end latency was derived as a function of the parameter $L$ and the packet erasure probability. In \cite{LowDelayFEC}, the feedback information was used to reduce the window size and hence reducing the coding complexity, but not to improve the latency performance. Another sliding-window network coding approach in conjunction with feedback-based selective repeat ARQ was proposed in \cite{Gabriel2018_2} to minimize the decoding delay and complexity. A joint scheduling and coding algorithm was recently proposed in \cite{Garrido2018}, which considered random packet arrivals. In \cite{Garrido2018}, a virtual queue was assumed at the receiver side, whose length is estimated at the transmitter side based on the delayed feedback information. Once the estimated queue length is above a certain threshold, a coded packet will be sent by the transmitter. In \cite{Sundararajan2017}, the authors investigated the design of FEC with instantaneous feedback for reducing the end-to-end latency in multicast networks. In \cite{Malak2019}, the investigated studied the FEC design to minimize the guaranteed delay, which is defined as the average delay plus three standard deviations. To reduce the complexity, the coding operation was limited to two consecutive packets in \cite{Malak2019}.

The aforementioned works \cite{Sundararajan2009,Cloud2015,Sundararajan2017,LowDelayFEC,Garrido2018,Gabriel2018_2,Malak2019} mainly focus on code constructions to reduce the coding complexity or achieve better latency performance than the existing benchmark schemes.  However, to our best knowledge, the \emph{optimal} end-to-end latency performance that can be achieved with delayed feedback\footnote{The scenario without any feedback can be treated as a special case of the delayed feedback by setting the feedback delay to infinity.}, together with the corresponding optimal FEC code construction, remain unknown.

\subsection{Our Contributions}
In this paper, we consider a communication system for sending a block of $N$ packets over a lossy channel with delayed feedback, to achieve the minimum average end-to-end latency $\bar{D}_{\mathrm{e2e}}$. The main contributions of the paper are summarized as follows:
\begin{itemize}
\item{First, we introduce a new rigorous mathematical model for the minimum-delay FEC design problem, which is formulated as a partially observable Markov decision process (POMDP). Specifically, the transmission process corresponds to the state space of the POMDP, which is characterized  by three parameters, i.e., the number of the transmitted information packets at the transmitter, the number of waiting packets to be sent to the application layer of the receiver and the number innovative packets that have been successfully received by the receiver for decoding the waiting packets. The coding decision corresponds to the action space of the POMDP, i.e., either sending a new information packet or sending a coded packet during each time slot. The average end-to-end latency corresponds to the accumulated reward of the POMDP. The delayed feedback information, if available, corresponds to the observation of the POMDP, which can be used to improve the estimation of the current state. The optimal FEC construction that achieves the minimum average end-to-end latency $\bar{D}_{\mathrm{e2e}}$ can be obtained by solving the corresponding POMDP.}
\item{Second, we propose a very efficient heuristic solution to the formulated POMDP, which is referred to as ``majority vote policy", where the action returned by majority vote of the states is selected at each step, with the voting of each state weighted by the probability for the environment being in that state. The FEC obtained by majority vote policy is shown to outperform the existing low-latency FEC designs \cite{LowDelayFEC}\cite{Garrido2018} for both scenarios without or with delayed feedback.  }
\item{Third, we derive tight lower and upper bounds for the optimal state values of the POMDP, based on which the ``$D$-step search algorithm" proposed in \cite{Ross2008} can be implemented for obtaining an improved policy relative to the simple majority vote policy.  With the $D$-step search, the current decision at each time step is made by looking $D$ steps ahead, which leads to asymptotically optimal performance as $D$ increases \cite{Ross2008}. Note that while theoretically the complexity of the $D$-step search algorithm grows exponentially with $D$, simulation results show that the close-to-optimal solution can be attained with $D=2$ for all our considered setups, due to the tight bounds of the optimal state values that we have derived.  }
\end{itemize}

The rest of this paper is organized as follows. Section~\ref{sec:POMDP} gives a brief overview of the general POMDP. In Section~\ref{sec:NoFB} and Section~\ref{sec:DelayedFB}, we respectively discuss the mathematical modeling and efficient algorithms for the minimum-latency FEC design for scenarios without and with delayed feedback. The performance of the proposed code designs is evaluated in Section~\ref{sec:simu} via extensive simulations, and finally we conclude this paper in Section~\ref{sec:con}.

\section{Partially Observable Markov Decision Process}\label{sec:POMDP}
This section gives a brief overview of the POMDP and introduces the key notations. We commence by introducing the fundamentals  of the Markov decision process (MDP) \cite{Smallwood1973}, which serves as a basis for the more complex POMDP. A MDP can be represented as a tuple $<\mathcal{S}, \mathcal{A},{T},{R}>$, where:
\begin{itemize}
\item{$\mathcal{S}$ is the state space of the environment.}
\item{$\mathcal{A}$ is the set of possible actions taken by the agent.}
\item{${T}:\mathcal{S}\times\mathcal{A}\times\mathcal S\rightarrow[0,1]$ is the state transition function, where $T(s,a,s')$ represents the probability of transiting to the next state $s'\in\mathcal{S}$ given that the current state is $s\in\mathcal{S}$ and action $a\in\mathcal{A}$ has been applied. }
\item{${R}:\mathcal{S}\times\mathcal{A}\rightarrow\mathbb{R}$ is the reward function, with $R(s,a)$ denoting the immediate reward for taking action $a$ at state $s$.}
\end{itemize}

The agent's actions are governed by the policy $\pi:\mathcal{S}\times\mathcal{A}\rightarrow[0,1]$, where $\pi(a|s)$ gives the probability of taking action $a$ when in state $s$. By solving a MDP, we aim to find the optimal policy, denoted as $\pi^*$, that maximize some measure of the long term reward. One commonly adopted measure of the long term return at time step $t$ is $G_t\triangleq\mathbb{E}\left[\sum_{k=0}^{\infty}\gamma^kR_{t+k}\right]$, where $R_{t+k}$ is the immediate reward at time step $t+k$ and $0<\gamma\leq 1$ is the discount factor for the future reward. The optimal value of a state $s\in\mathcal{S}$, denoted as $v^*(s)$, is defined as the average return starting from state $s$ and following optimal policy $\pi^*$ thereafter.
The optimal state-value function is the fixed point of the celebrated Bellman's equation \cite{Smallwood1973}, i.e.,
\begin{align}
v^{*}(s)=\max_{a\in\mathcal A}\left[R(s,a)+\gamma\sum_{s'\in \mathcal S}T(s,a,s')v^*(s')\right],\forall s\in\mathcal S.
\end{align}

In MDP, the agent directly observes the current state, based on which an action is taken according to the policy. In contrast, POMDP deals with partially observable environment, where the current state cannot be directly observed. Instead, the agent can only access some observations which give incomplete information for the current state. Specifically, a POMDP can be described by a tuple $<\mathcal{S},\mathcal{A},\mathcal{T},\mathcal{R},\mathcal{Z},{O}>$, where
\begin{itemize}
\item{$<\mathcal{S}, \mathcal{A},\mathcal{T},\mathcal{R}>$ defines a MDP.}
\item{$\mathcal{Z}$ is the set of observations that the agent can access.}
\item{${O}:\mathcal{S}\times\mathcal{A}\times\mathcal{Z}\rightarrow [0,1]$ is the observation function, where $O(s,a,z)$ gives the probability for observing $z\in\mathcal{Z}$ after action $a\in\mathcal{A}$ is taken at the state $s\in\mathcal{S}$. }
\end{itemize}

Since the states are not observable in POMDP, the agent has to choose its actions based on the complete history of past observations and actions, which can be quite memory-expensive. It was shown in \cite{Smallwood1973} that it is sufficient to summarize all the information to a belief of the current state, which is referred to as ``belief state". The belief state at time $t$ is defined as the probability distribution of being in each state $s\in\mathcal{S}$, i.e.,
\begin{align}
b_t(s)=\Pr\left(s_t=s|\{a_i,i=0,...,t-1\},\{z_i,i=1,...,t\},b_0\right),\forall s\in\mathcal{S},
\end{align}
where $a_i$ and $z_i$ are the action and observation at time $i$, respectively, and $b_0$ is the initial belief state. At any time $t$, the belief state $b_t$ can be computed from the previous belief state $b_{t-1}$, the previous action $a_{t-1}$ and the current observation $z_t$. Specifically, the belief update function, denoted as $b_t=\tau(b_{t-1},a_{t-1},z_t)$, can be written as
\begin{align}
b_t(s)=\frac{O(s,a_{t-1},z_t)}{\Pr(z_t|b_{t-1},a_{t-1})}\sum_{s'\in\mathcal S}T(s',a_{t-1},s)b_{t-1}(s'), \forall s\in\mathcal{S},
\end{align}
where $\Pr(z|b,a)\triangleq\sum_{s\in\mathcal{S}}O(s,a,z)\sum_{s'\in\mathcal S}T(s',a,s)b(s')$ is the probability of observing $z$ after taking action $a$ at belief $b$.

The policy of a POMDP specifies the probability for choosing each action under any given belief state. The optimal value of belief state $b$, denoted as $v^*(b)$, is defined as the return received by following the optimal policy $\pi^*$, based on the belief state $b$. It is related with the optimal value of states as
\begin{align}
v^*(b)=\sum_{s\in\mathcal S}b(s)v^*(s), \label{eq:ValueBelief}
\end{align}
where $v^*(s)$ is the optimal state value of $s$, for $s\in\mathcal{S}$. The optimal belief-value function is the fixed point of the following Bellman's equation \cite{Smallwood1973}
\begin{align}
v^*(b)=\max_{a\in\mathcal{A}}\left[\sum_{s\in\mathcal{S}}b(s)R(s,a)+\gamma\sum_{z\in\mathcal Z}O(s,a,z)v^*(\tau(b,a,z))\right], \label{eq:bellman}
\end{align}
where $\tau(b,a,z)$ is the updated belief after taking action $a$ at belief state $b$ and observe $z$.

The optimal solution to the general POMDP can be found through the value iteration algorithm \cite{Sondik1971}. However, it is generally computationally infeasible to solve the POMDP optimally unless the problem size is very small \cite{Kaelbling1998}. Hence, in many practical problems, we are interested in obtaining heuristic solutions that can be solved efficiently. The most straightforward heuristic solution is to simply ignore the uncertainty of the states and choose the optimal MDP actions for the most likely state, i.e.,
\begin{align}
\pi_{MLS}(b)=\pi^{*}_{MDP}\left(\arg\max_{s}b(s)\right), \label{eq:huerstic1}
\end{align}
where $\pi^{*}_{MDP}(s)$ is the optimal policy of the corresponding fully observable MDP. However, this solution usually has very poor performance for problems with high uncertainty of the states. A more sophisticated approximation technique was introduced in \cite{Littman1995}, which chooses the action that maximizes the expected return given the current belief by assuming that the POMDP will become fully observable after taking this subsequent action. The corresponding policy is given by
\begin{align}
\pi_{MDP}(b)=\arg\max_{a\in\mathcal A}\sum_{s\in\mathcal S}b(s)\sum_{s'\in\mathcal S}\gamma T(s,a,s')v^*_{MDP}(s'),\label{eq:MDPsolu}
\end{align}
where $v^*_{MDP}(s')$ is the optimal value of state $s'$ in the fully observable MDP. Note that this policy assumes that there is no uncertainty over the state after taking one action, and hence it usually performs poorly in the scenarios where information gathering about the environment state is important, i.e., taking the actions that leads to less uncertainty of the states. Many other heuristic solutions to a general POMDP that achieve different tradeoffs between the complexity and performance can be found in \cite{Ross2008}.

For the minimum-latency FEC design problem considered in this paper, both heuristic solutions in \eqref{eq:huerstic1} and \eqref{eq:MDPsolu} perform poorly. To this end, we will first introduce an equally-efficient heuristic solution as compared with \eqref{eq:huerstic1} and \eqref{eq:MDPsolu}, named ``majority vote policy", which renders a significantly better latency performance. Furthermore, we will also derive tight lower and upper bounds for the optimal state values of the POMDP corresponding to the minimum-latency FEC design problem, based on which a more sophisticated $D$-step search algorithm can be implemented for obtaining near-optimal solutions.

\section{Minimum-Latency FEC Design without Feedback}\label{sec:NoFB}
In this section, we consider the design of FEC for sending a block of $N$ packets in the scenario without feedback from the receiver, aiming at achieving the minimum average end-to-end latency. In the absence of feedback, the propagation delay can be counted as a constant added to the average end-to-end latency $\bar{D}_{\mathrm{e2e}}$ in \eqref{eq:expectedDelay}, which does not affect the code design and hence it is neglected.

For simplicity, we assume that all the packets are generated at $t=0$, i.e., at the beginning of the first time slot. The proposed code design also applies to the scenario where the packets are sequentially generated at consecutive time slots, where the queueing delay for each packet is simply reduced by a constant value from the one for the case of all packets being generated at $t=0$. Denote by $t(i), i=1,..,N$, the time that the $i$th packet is delivered to the application layer at the receiver, i.e., the $i$th packet is delivered by the end of the time slot $t(i)$. Then, the end-to-end latency experienced by the $i$th packet is $D_{\mathrm{e2e}}(i)=t(i)$ and the average end-to-end latency in \eqref{eq:expectedDelay} is equivalently expressed as
\begin{align}
\bar{D}_{\mathrm{e2e}}=\frac{1}{N}\sum_{i=1}^{N}t(i). \label{eq:expecteDelayNoFB}
\end{align}

We assume that these $N$ packets are sent through an independent and identically distributed (i.i.d.) packet erasure channel, where a typical packet is erased with probability $p$ and received with probability $1-p$, $0<p<1$. Without loss of optimality for achieving the minimum end-to-end latency, we consider the code construction with the following assumptions:
 \begin{itemize}
 \item{The information packets are transmitted sequentially.}
 \item{At deliberately chosen time slots, the coded packets are sent for correcting the packet erasures, which are generated via RLNC from all the processed information packets in previous time slots.}
 \item{The transmission terminates when all the packets are successfully delivered to the application layer of the receiver.}
 \end{itemize}
Denote by $a_i$ the coding decision (i.e., the action) at time slot $i$, where
\begin{align}
a_i=\begin{cases}
0, & \textnormal{if an information packet is sent at time slot $i$,}\\
1, & \textnormal{if a coded packet is sent at time slot $i$}.
\end{cases}
\end{align}
At the $1$st time slot, it is obvious that we can only send the information packet, i.e., $a_1=0$. Since there are in total $N$ information packets, we have $\sum_{i=0}^{\infty}(1-a_i)=N$, i.e., action $a=0$ is chosen for exactly $N$ times.

In the absence of feedback, the code design can be performed off-line. Hence, the minimum-latency FEC can be obtained using brute-force search. Assume that the maximum number of coded packets prior to the $N$th information packet transmission does not exceed $K$. The total number of coding policies that we need to examine is ${N+K}\choose {K}$. If the channel erasure probability is $p$, a reasonable value of $K$ should be slightly greater than $\lceil\frac{Np}{1-p}\rceil$. The brute force search is computational prohibitive when $N$ is large. For example, for sending 20 information packets over a channel with erasure probability 0.3, the total number of coding policies to be evaluated is more than $3\times 10^7$.

\subsection{Mathematical Modeling}\label{sec:math_noFB}
In this subsection, we propose to model the minimum-latency FEC design problem as a POMDP, based on which efficient algorithms can be applied to obtain the optimal coding decisions. Each element of the POMDP, as discussed in Section~\ref{sec:POMDP}, will be elaborated in the following.

First, it is noted that the coding decision is  affected by the number of waiting packets at the receiver side, denoted as $w$, which is the number of information packets that have been sent at the transmitter side, but have not yet been delivered to the application layer of the receiver. This corresponds to the case when either the packet itself or its previous packet cannot be successfully decoded. If there is no waiting packet, i.e., $w=0$, it is obvious that the transmitter should immediately send the next information packet in the queue, i.e., $a_i=0$. On the other hand, when there are some packets waiting to be delivered, i.e., $w\geq 1$, coded packets should be sent by the transmitter to enable the receiver to decode the waiting packets as soon as possible. The number of coded packets that are required for successfully decoding all the waiting packets depends on the number of {innovative} packets already available at the receiver, which is denoted as $d$. With RLNC, the packets are decoded once the number of innovative packets collected by the receiver is equal to the number of waiting packets, i.e., $d=w$. Furthermore, the coding decision at the transmitter also depends on the number of packets that have already been sent by the transmitter, denoted as $n$. This is because the total increased queueing delay during a time slot is related to the number of packets waiting at the transmitter, i.e., $N-n$.  Based on the above discussions,  the states of the corresponding POMDP are characterized  by the tuple $s\triangleq(n,w,d)$, where $n,w$ and $d$ denote the number of transmitted information packets, the number of waiting packets and the number of innovative packets at the receiver, respectively, with $0\leq d<w\leq n\leq N$. The total number of states of this POMDP is thus given by
\begin{align}
|\mathcal S|=\sum_{n=0}^{N}\sum_{w=0}^{n}\sum_{d=0}^{w-1}1=N+\sum_{n=1}^{N}\frac{n(n+1)}{2}=N+\frac{N(N+1)(N+2)}{6}.
\end{align}

Next, we consider the action space and the reward function. The set of actions during each time slot is $\mathcal{A}=\{0,1\}$, i.e., sending an information packet when $a=0$ and a coded packet when $a=1$. In order to minimize the average end-to-end latency, we define the immediate reward for each action as the negative of the increased latency in $N\cdot \bar{D}_{\mathrm{e2e}}$ during that time slot, which is equivalent to the number of packets that have not yet been delivered to the application layer of the receivers at that time slot\footnote{If a packet is delivered by the end of a time slot, this time slot still counts into the latency of this packet, as illustrated in Fig.~\ref{F:delay}.}. In other words, the average end-to-end latency given in \eqref{eq:expecteDelayNoFB} is related with the accumulated reward as
\begin{align}
\bar{D}_{\mathrm{e2e}}=-\frac{1}{N}\sum_{t}R(s_t,a_t).
\end{align}
Thus, at state $s=(n,w,d)$, we have $R(s,a)=-(N-n+w), \forall a\in\mathcal{A}$, which includes the $(N-n)$ packets waiting in the queue at the transmitter side and the $w$ packets waiting in the buffer at the receiver side. Note that for each given state, while the action $a$ does not directly affect the immediate reward, it affects the state transition  (as given below) and hence will influence the reward indirectly for the next time slot.

At the initial state $s=(0,0,0)$, only an information packet can be sent. After all the information packets are sent, the system will be in the state $(N,\cdot,\cdot)$ and only coded packets can be sent for all subsequent time slots. Thus, for $n=0$ or $n=N$, the optimal action is trivially solved. On the other hand, for $0<n<N$, we can either send an information packet or a coded packet.  The state transition depends on the action applied and the channel realization, i.e., whether the packet is received or erased.  Specifically, if an information packet is sent, i.e., $a=0$, we have the following situations:
\begin{itemize}
\item{If $w=0$, there is no waiting packet at the receiver. In this case, if the information packet sent at the current time slot is received successfully, which has the probability of $1-p$, the state will transit to the next state $(n+1,0,0)$. On the other hand, if this information packet is erased, which has the probability of $p$, the next state will be $(n+1,1,0)$. }
\item{If $w>0$, some preceding packets are not decodable. In this case, if the information packet sent at the current time slot is received successfully, which has the probability of $1-p$, both the number of waiting packets and the number of innovative packets at the receiver are increased by 1, i.e., the state will transit to $(n+1,w+1,d+1)$. On the other hand, if this packet is erased, the state will transit to $(n+1,w+1,d)$. }
\end{itemize}
On the other hand, if a coded packet is sent, i.e., $a=1$, we have the following situations:
\begin{itemize}
\item{If $d<w-1$, there is no chance of decoding the waiting packets regardless whether the packet sent by the current time slot is received successfully or not. Furthermore, since there is no new information packet processed, the number of processed packets and the number of waiting packets will not change. If this coded packet is received successfully, the state will transit to $(n,w,d+1)$, otherwise it will remain at $(n,w,d)$. }
\item{If $d=w-1$, we can decode all the waiting packets if this coded packet is received successfully, for which the state will be transited to $(n,0,0)$. If this coded packet is erased, the state will remain at $(n,w,d)$.   }
\end{itemize}
The transition of states for this POMDP is summarized in Fig.~\ref{F:transition}. When all the packets are successfully delivered, the transition diagram terminates at the final state $(N,0,0)$.
\begin{figure}[htb]
\centering
\includegraphics[scale=0.7]{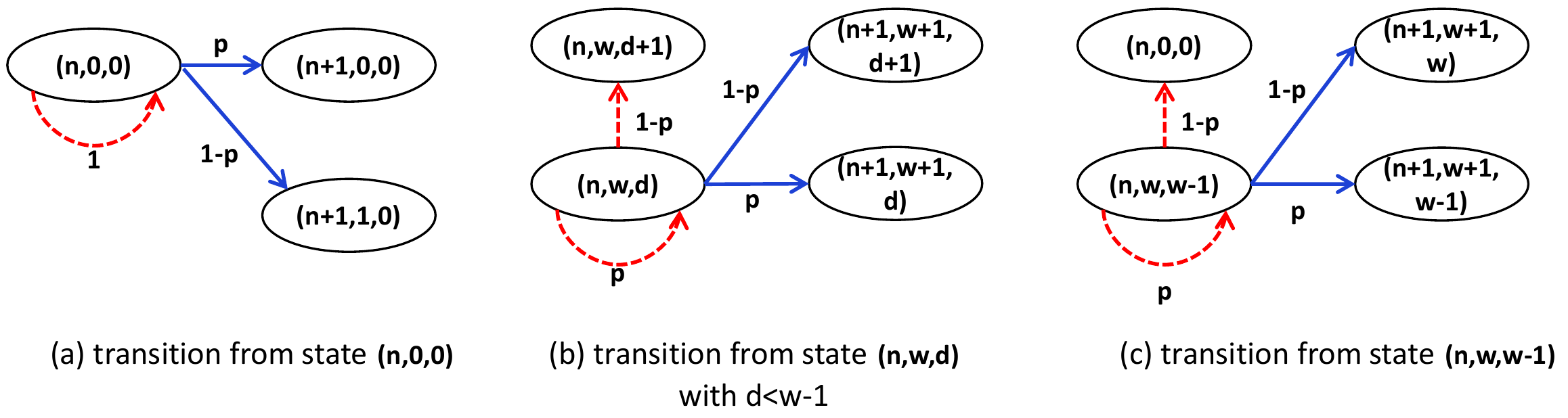}
\caption{An illustration of the state transitions. The blue lines denote transitions after taking action $a=0$ and the red dashed lines denote transitions after taking action $a=1$, with transition probabilities labeled on the lines. }
\label{F:transition}
\end{figure}

As a result, we model the minimum-latency code design as a general MDP, which includes
\begin{itemize}
\item{The state space $\mathcal{S}=\{(n,w,d),0\leq d<w\leq n\leq N\}$,}
\item{The action space $\mathcal{A}=\{0,1\}$,}
\item{The state transition function
\begin{equation}\label{eq:transitionFun}
\begin{aligned}
&T((n,w,d),0,(n',w',d'))=
\begin{cases}
1-p, & \textnormal{if $n'=n+1,w'=w=0$ or $n'=n+1,w'=w+1$, }\\
p, &\textnormal{if $n'=n+1,w'=w+1,d'=d$,}\\
0, &\textnormal{otherwise}
\end{cases}\\
&T((n,w,d),1,(n',w',d'))=
\begin{cases}
1-p, & \textnormal{if $n'=n,w'=w,d'=d+1$ or $n'=n,d=w-1,w'=0$, }\\
p, &\textnormal{if $n'=n,w'=w,d'=d$,}\\
0, &\textnormal{otherwise}
\end{cases}
\end{aligned}
\end{equation}
}
\item{The reward function $R((n,w,d),a)=-(N-n+w),\forall a\in\mathcal{A}$.}
\end{itemize}

We refer to the above MDP by $\mathrm{MDP}_{1}$ in the following context. Since the states in $\mathrm{MDP}_{1}$ are not directly observable at the transmitter, this problem is classified as POMDP. Note that the transmitter can directly observe the number of transmitted  packets in $\mathrm{MDP}_{1}$, i.e., the parameter $n$ is known exactly by the transmitter. On the other hand, the parameters $w$ and $d$ correspond to the status of the receiver, which remains unknown for the transmitter in the absence of feedback. The transmitter can estimate the joint distribution of $w$ and $d$ based on the historical actions. Specifically, when $n$ information packets are transmitted, there are in total $\frac{n(n+1)}{2}+1$ possible states, which are
\begin{align}
\mathcal{S}_n&=\{(n,w,d), 0\leq d<w\leq n\}\nonumber\\
&=\{(n,0,0),(n,1,0),(n,2,1)(n,2,0),\cdots,(n,w,w-1),\cdots,(n,w,0),\cdots,(n,n,0)\}.
\end{align}
At the beginning of the $t$th time slot, if $n$ information packets have been sent, we represent the belief state $b_t$ by a compact vector $\mathbf{b}_t^n\in\mathbb{R}^{1\times\left(\frac{n(n+1)}{2}+1\right)}$. For notational convenience, we denote by $\mathbf{b}_t^n(w,d), 0\leq d<w\leq n$, the probability for being in state $(n,w,d)$.

At the beginning of the $1$st time slot, the system starts from the initial state $(0,0,0)$ and the initial belief is $\mathbf{b}_1^0=[1]$. Then, based on the state transition shown in Fig.~\ref{F:transition}, the belief vector $\mathbf{b}_t^n$ is updated according to the following rules:
\begin{itemize}
\item{If an information packet is sent, i.e., $a=0$, we have new belief vector $\mathbf{b}_{t+1}^{n+1}$ given by
    \begin{equation}\label{eq:beliefUpdate1}
    \begin{aligned}
    &\mathbf{b}_{t+1}^{n+1}(0,0)=(1-p)\mathbf{b}_t^n(0,0),\\
    &\mathbf{b}_{t+1}^{n+1}(w,0)=p\mathbf{b}_t^n(w-1,0),\forall w\geq 1,\\
    &\mathbf{b}_{t+1}^{n+1}(w,d)=(1-p)\mathbf{b}_t^n(w-1,d-1)+p\mathbf{b}_t^n(w-1,d). \  \forall w\geq2\textnormal{ and } d<w-1,\\
    &\mathbf{b}_{t+1}^{n+1}(w,w-1)=(1-p)\mathbf{b}_t^n(w-1,w-2),\ \forall w\geq 2.
    \end{aligned}
    \end{equation}
    }
\item{If a coded packet is sent, i.e., $a=1$, the belief vector $\mathbf{b}_t^n$ is updated as
    \begin{equation}\label{eq:beliefUpdate2}
    \begin{aligned}
    \mathbf{b}_{t+1}^n(0,0)&=\mathbf{b}_t^n(0,0)+(1-p)\sum_{w=1}^{n}\mathbf{b}_t^n(w,w-1),\\
    \mathbf{b}_{t+1}^n(w,0)&=p\mathbf{b}_t^n(w,0),\forall w\geq 1,\\
    \mathbf{b}_{t+1}^n(w,d)&=(1-p)\mathbf{b}_t^n(w,d-1)+p\mathbf{b}_t^n(w,d).\ \forall w\geq 1.
    \end{aligned}
    \end{equation}
}
\end{itemize}

\begin{example}
Consider the simple case where only two packets need to be sent, i.e., $N=2$. The corresponding MDP is illustrated in Fig.~\ref{F:example}. At the initial state $(0,0,0)$, we send out the first information packet. Following that, we can choose either to send a repetition of the first packet (considered as a coded packet) or send the second information packet. After two information packets have been sent, the system will be in the state $(2,\cdot,\cdot)$ from which only action $a=1$ can be applied, i.e., sending a linear combination of these two packets. The probability and the immediate reward are labelled with each state transition. The belief state can be updated based on the historical actions according to \eqref{eq:beliefUpdate1} and \eqref{eq:beliefUpdate2}. For example, after sending the first information packet from the initial state $(0,0,0)$, we have the belief vector $\mathbf{b}_2^1=[1-p,p]$, where $p$ is the packet erasure probability. Following from that, if a coded packet is sent (basically the repetition of the first information packet), the belief vector is updated as $\mathbf{b}_3^1=[1-p^2,p^2]$. On the other hand, if the second information packet is sent immediately after the first one, the belief vector is $\mathbf{b}_3^2=[(1-p)^2,p(1-p),p(1-p),p^2]$.
\begin{figure}[htb]
\centering
\includegraphics[scale=0.8]{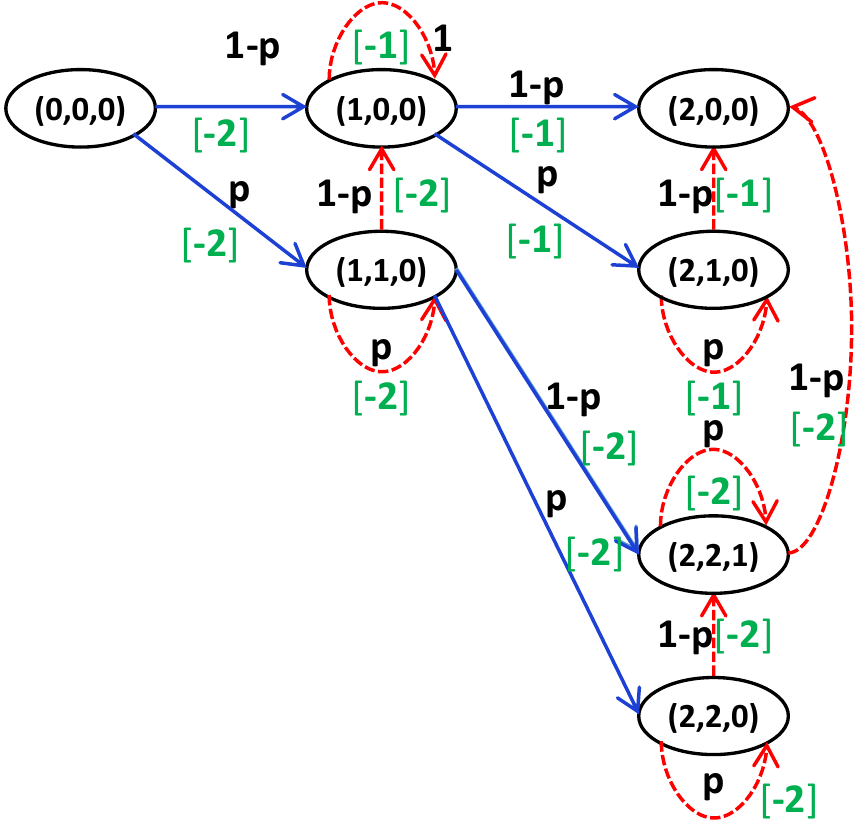}
\caption{The state transition diagram for a minimum-latency code design with $N=2$. The blue lines denote transitions after taking action $a=0$ and the red dashed lines denote transitions after taking action $a=1$. The transition probabilities and the immediate reward associated with each transition are labeled on the lines. }
\label{F:example}
\end{figure}
\end{example}


\subsection{Efficient Algorithms}\label{sec:alg_noFB}
The optimal solution to the general POMDP can be found through the value iteration algorithm \cite{Sondik1971}. However, the number of policies that need to be examined in the value iteration algorithm grows exponentially with the number of packets to be delivered, and hence the complexity for implementing this optimal algorithm is prohibitive unless $N$ is very small, e.g., $N<10$. Meanwhile, we note that those heuristic policies discussed in Section~\ref{sec:POMDP} lead to a poor performance for our considered problem. Hence, in this subsection, we propose novel efficient algorithms that can achieve near-optimal solutions for the considered POMDP.

First, it is noted that if the perfect knowledge of the state is available, the optimal action is straightforward, i.e., an information packet should be sent whenever there is no waiting packet at the receiver side, i.e., when $w=0$; while a coded packet should be sent as long as  there is a waiting packet i.e., when $w>0$. Hence, a straightforward solution for the considered POMDP is to choose the most likely optimal action based on the belief state, which is referred to as \emph{majority vote policy}. Denote by $\mathcal{S}_a\triangleq\{s:\pi^*_{MDP}(s)=a\}$ the set of states for which the optimal action is given by $a\in\mathcal A$ in the fully observable MDP. Then, the majority vote policy for the general POMDP is defined as
\begin{align}
\pi_{MV}(b)=\underset{a\in\mathcal{A}}{\arg\max}\sum_{\mathcal{S}_a}b(s).\label{eq:MVGeneral}
\end{align}
Note that \eqref{eq:MVGeneral} is an alternative maximum likelihood solution similar to \eqref{eq:huerstic1}. The solution in \eqref{eq:huerstic1} choose the optimal action that maps to the most likely state, i.e., one state, while \eqref{eq:MVGeneral} choose the action that is optimal for most of the states, weighted by the state probabilities. Hence, the solution in \eqref{eq:MVGeneral} usually provides a better performance when the action space is small

For the minimum-latency code design problem considered in this paper, the action space is $\mathcal A=\{0,1\}$. The majority vote policy can be equivalently expressed as
\begin{align}
\pi_{MV}(\mathbf{b}_t^n)=
\begin{cases}
0, & \textnormal{if $\mathbf{b}_t^n(0,0)>0.5$}\\
1, & \textnormal{otherwise}.
\end{cases}\label{eq:MVPolicy}
\end{align}
The majority vote policy focuses on the local optimal action at each time instance, while ignoring the effect of the current action on the future states. The performance can be further enhanced by looking a few steps ahead, by adopting the \emph{$D$-step search} algorithm proposed in \cite{Ross2008}. With the $D$-step search algorithm, a tree of reachable belief state from the current belief state is built by examining all the possible sequences of $D$ actions/observations\footnote{For the low-latency code design problem, the observation is available for the case with delayed feedback. In this section, we focus on the scenario without feedback, where the observation space is empty. The low-latency code design with delayed feedback is considered in Section~\ref{sec:DelayedFB}.} that can be taken/observed from the current belief. The belief nodes are represented using logical $\mathrm{OR}$-nodes, at which we choose one action from the action space, while the actions are included in between each layer of belief nodes using logical $\mathrm{AND}$-nodes, since all the possible subsequent observations must be considered. Then, to make the decision based on the root belief node, we first estimate the values of the belief nodes at the fringe of this tree using an approximated value function computed offline. Then, the state action values of the upper layers are estimated based on the value iterations defined in \eqref{eq:bellman}. Specifically, in the absence of feedback, the value iteration function is simplified to
\begin{align}
v(b_t)=\max_{a\in\mathcal{A}}\left[\sum_{s\in\mathcal{S}}b_t(s)R(s,a)+v(\tau(b_t,a,z))\right],
\end{align}
where $\tau(b_t,a,\cdot)$ is the belief update function without observation. The most important task for implementing the $D$-step search algorithm is to obtain a good approximated value function of the fringe belief nodes.

First, it is noted that for the minimum-latency code design problem, there is only one possible action after all the $N$ information packets are sent at the transmitter, i.e., sending the coded packets. In this case, the optimal value of the states $(N,w,d)$ can be evaluated according to Lemma~\ref{lem:finalValue} given below. For notational convenience, we denote the value function of the state $v(s)$ for $s=(n,w,d)$ by $v(n,w,d)$, when there is no ambiguity.
\begin{lemma}\label{lem:finalValue}
For the minimum-latency code design problem with $N$ packets, the optimal value of the state $(N,w,d)$ is $v^*(N,w,d)=-\frac{w(w-d)}{1-p}$ for $0\leq d<w\leq N$, where $p$ is the packet erasure probability of the channel.
\end{lemma}
\begin{proof}
Please refer to Appendix~\ref{A:finalValue}.
\end{proof}
Given the optimal value function of all the states $\{(N,w,d),0\leq d<w\leq N\}$, we can obtain the optimal value of the belief $\mathbf{b}_t^{N}$ according to \eqref{eq:ValueBelief}. However, the optimal value function for the belief states $\mathbf{b}_t^{n}$ with $n<N$ remains unknown. To find the approximated value functions for the general belief state, we consider the upper and lower bounds for the optimal value function.


Note that a lower bound $v^{L}(b_t)$, where $v^*(b_t)\geq v^{L}(b_t)$ of the optimal belief value can be obtained based on the majority vote policy given in \eqref{eq:MVPolicy}. Given the belief state $b_t$, all the subsequent actions and belief states can be determined based on the majority vote policy, until reaching the belief state $b_{t'}^{N}$, whose optimal value is known from Lemma~\ref{lem:finalValue}. Mathematically, the lower bound of the belief value can be calculated recursively as
\begin{align}
v^{L}(b_t)=\sum_{s\in\mathcal{S}}b_t(s)R(s,a)+v(b_{t+1}),\label{eq:ValueLB}
\end{align}
where $a=\pi_{MV}(b_t)$ is the action obtained by applying majority vote policy and $b_{t+1}=\tau(b_{t},a,\cdot)$ is the subsequent belief.

On the other hand, an upper bound of the optimal belief value $v^{U}(b_t)$, where $v^*(b_t)\leq v^{U}(b_t)$, can be obtained by assuming that the states in $\mathrm{MDP}_{1}$ are fully observable and the optimal action is taken at every step. Then, the optimal value of the state is given in Lemma~\ref{lem:OptValues} below.
\begin{lemma}\label{lem:OptValues}
If $\mathrm{MDP}_{1}$ is fully observable, the optimal value of state $(n,w,d),\forall  0\leq d<w\leq n\leq N$, is given by
\begin{align}
v_{MDP}^{*}(n,w,d)=-\frac{(N-n+w)(w-d)}{1-p}-\frac{(N-n+1)(N-n)}{2(1-p)}.\label{eq:OptValue}
\end{align}
\end{lemma}
\begin{proof}
Please refer to Appendix~\ref{A:OptValue}.
\end{proof}

Then, the upper bound of the belief value $v^{U}(b_t)$ can be calculated as
\begin{align}
v^{U}(b_t)=\sum_{s\in\mathcal S}b_t(s)v_{MDP}^{*}(s),\label{eq:ValueUB}
\end{align}
where $v_{MDP}^{*}(s)$ is the optimal value of state $s$ in the fully observable MDP as given in \eqref{eq:OptValue}.

With the $D$-step search algorithm, we calculate the value of the fringe nodes of the belief tree offline, e.g., the lower and upper bound values of the fringe nodes can be calculated from \eqref{eq:ValueLB} and \eqref{eq:ValueUB}, respectively. Then, these bounds are propagated to the parent nodes according to the following equation
\begin{align}
&L(b)=\begin{cases}
v^{L}(b), & \textnormal{if $b$ is a fringe node}\\
\max_{a\in\mathcal{A}}L(b,a), & \textnormal{otherwise},
\end{cases}\label{eq:lowerBound}\\
&\textnormal{with }L(b,a)\triangleq\sum_{s\in\mathcal S}b(s)R(s,a)+L(\tau(b,a,\cdot));\\
&U(b)=\begin{cases}
v^{U}(b), & \textnormal{if $b$ is a fringe node}\\
\max_{a\in\mathcal{A}}U(b,a), & \textnormal{otherwise},
\end{cases}\label{eq:upperBound}\\
&\textnormal{with }U(b,a)\triangleq\sum_{s\in\mathcal S}b(s)R(s,a)+U(\tau(b,a,\cdot)),\label{eq:beliefActionValue}
\end{align}
where $R(s,a)$ is the immediate reward for implementing action $a$ at state $s$, and  $\tau(b_t,a,\cdot)$ is the updated belief with action $a$ taken on the current belief $b_t$.

After obtaining the upper and lower bounds for the current belief state $b_t$, we can take the action that leads to the maximal approximated belief value. In this paper, we approximated the optimal value of $b_t$ with the lower bound $L(b_t)$. Hence, the policy can be expressed as
\begin{align}
\pi_{D-\mathrm{step}}(b_t)=\underset{a\in\mathcal A}{\arg\max}\ L(b_t,a).\label{eq:DStepPolicy}
\end{align}

Note that the size of the belief tree increases exponentially with the number of steps we look ahead, and hence the complexity of the $D$-step search algorithm increases exponentially with $D$. The complexity for implementing $D$-step search algorithm is equivalent to evaluate $2^D$ policies, which is far less than the brute search algorithm, which needs to evaluating ${N+K}\choose{K}$ policies, e.g., for $N=20$ and $p=0.3$, the $4$-step search algorithms is equivalent to evaluate $16$ policies while the brute force search needs to evaluate over $3\times 10^7$ policies. With the lower bound and upper bound obtained in \eqref{eq:lowerBound}-\eqref{eq:beliefActionValue}, we can reduce the size of the belief tree by applying the classical branch-and-bound pruning algorithm \cite{Ross2008}. Specifically, for a given belief state, if the upper bound for taking action $a$ is lower than the lower bound for taking another action $a'$, i.e., $U(b,a)\leq L(b,a')$, then action $a$ is strictly suboptimal in this belief state, and hence that branch (action $a$ and the subsequent reachable belief states) can be pruned without performance degradation. The pseudocode for general Branch-and-Bound pruning algorithm can be found in \cite{Ross2008} (Algorithm 3.2 on page 12).

\begin{example}
Consider the simple case with $N=2$. The state transition diagram for the minimum latency code design problem is shown in Fig.~\ref{F:example}. After the first information packet is sent, the belief state is $b_2=\mathbf{b}_2^1=[1-p,p]$, where $p$ is the erasure probability. For the second time slot, we can either send the next information packet, i.e., with $a=0$, or send a repetition of the first information packet, i.e., with $a=1$. With majority vote, the action $a=1$ should be taken if $p>0.5$; otherwise, action $a=0$ should be taken.

With the 2-step search algorithm, the decision can be made by building a belief tree as shown in Fig.~\ref{F:tree_noFB}, by considering all the possible actions and reachable belief states in the subsequent two steps. The fringe nodes $\mathbf{b}_3^2$ and $\mathbf{b}_4^2$ correspond the the set of states $\{(N,w,d), 0\leq d<w\leq N\}$, whose optimal values are directly obtained from Lemma~\ref{lem:finalValue}. Hence, we can calculate the optimal value of the belief $\mathbf{b}_3^2$ and $\mathbf{b}_4^2$ as $v^*(\mathbf{b}_3^2)=-3p-\frac{4 p^2}{1-p}$ and $v^*(\mathbf{b}_4^2)=-p-3p^2-\frac{4p^3}{1-p}$. Meanwhile, for the value of $\mathbf{b}_4^1$, we can calculate the upper and lower bound from \eqref{eq:ValueUB} and \eqref{eq:ValueLB}, respectively, which are\footnote{It is generally difficult to give the explicit expression for \eqref{eq:ValueLB} as a function of $p$, except in this simple case with $N=2$. However, the numerical value of the lower bound can be calculated very efficiently according to \eqref{eq:ValueLB}, with complexity $\mathcal{O}(N)$.}
\begin{align}
v^{U}(\mathbf{b}_4^1)=-\frac{1+2p^3}{1-p};\  v^{L}(\mathbf{b}_4^1)=
\begin{cases}
-1-3p^3-\frac{p(1-p^3)}{1-p}-\frac{4p^4}{1-p}, & \textnormal{if $p^3<0.5$}\\
-\frac{p(1-p^m)}{1-p}-2p^m-\frac{4p^{m+1}}{1-p}-\sum_{i=3}^{m}(1+p^i), & \textnormal{otherwise}
\end{cases}.
\end{align}
where $m=\underset{u}{\arg\min}(p^u<0.5)$.

The lower and upper bounds can be propagated to the parent $\mathbf{b}_3^1$ according to \eqref{eq:lowerBound} and \eqref{eq:upperBound}. For the example with $p=0.3$, the values for each node are labeled in Fig.~\ref{F:tree_noFB}. Under the belief $\mathbf{b}_2^1$, action $a=1$ is strictly suboptimal than the action $a=0$, and hence the branch from that action node can be pruned. Then, we have $L(\mathbf{b}_3^1)=U(\mathbf{b}_3^1)=-1.09+v^*(\mathbf{b}_4^2)=-1.8143$. Based on the belief tree, we have $L(\mathbf{b}_2^1,0)>L(\mathbf{b}_2^1,1)$, which implies that action $a=0$ should be taken at $\mathbf{b}_2^1$. For this simple example with $N=2$ and $p=0.3$, we have $\pi_{MV}=\pi_{D-\mathrm{step}}$ and both of them are optimal. However, for general $N$ and $p$, the $D$-step search policy usually outperforms the majority vote policy, as shown in Section~\ref{sec:simu}.
\begin{figure}[htb]
\centering
\includegraphics[scale=0.55]{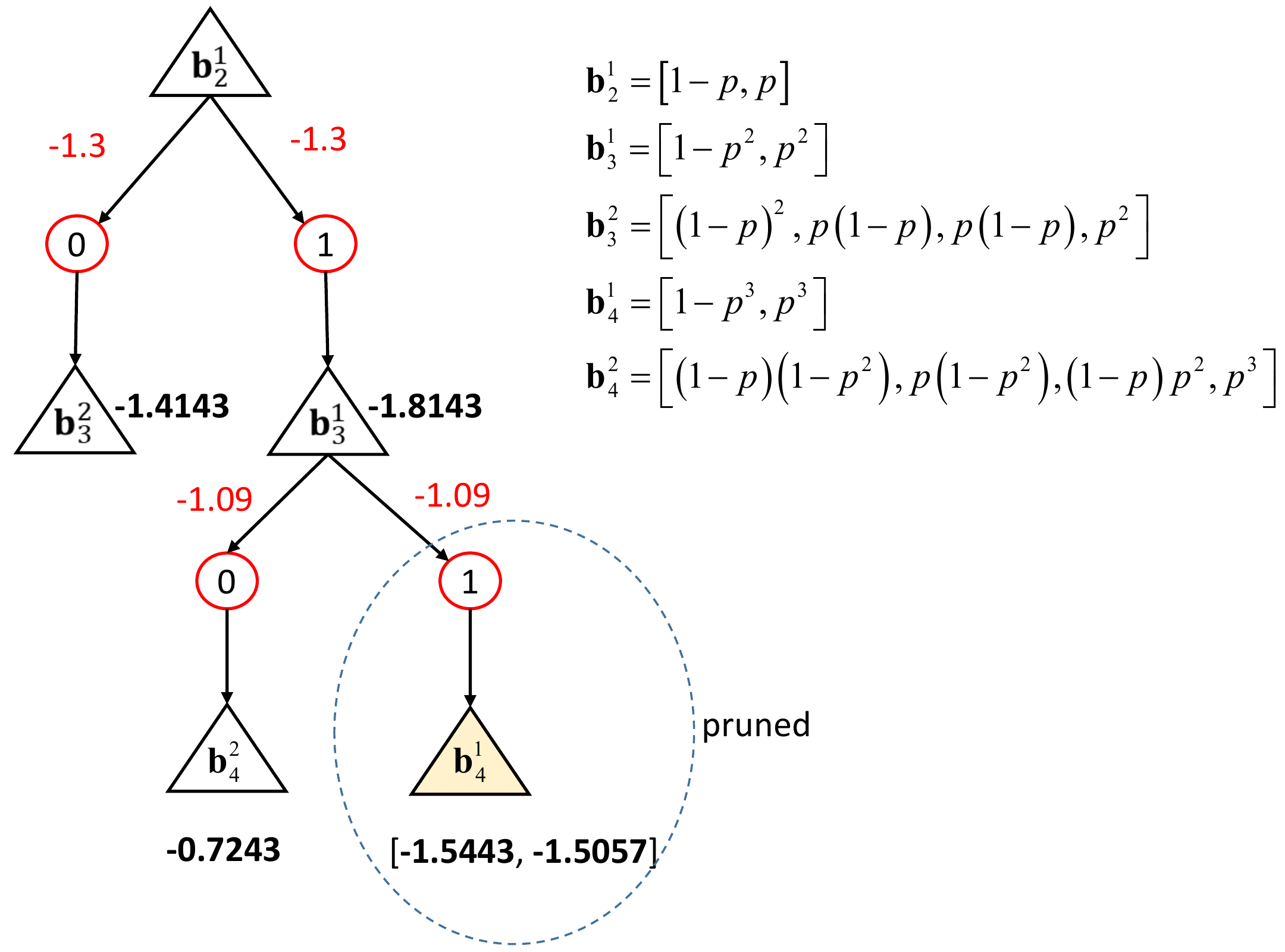}
\caption{The 2-step belief tree for $N=2$ at $t=2$. The triangles denote the belief states and the circles denote the actions. For $p=0.3$, the expected immediate reward for belief-action pairs are labeled with red text along the lines, and the (range of the) optimal belief state values are labeled with bold black text.    }
\label{F:tree_noFB}
\end{figure}
\end{example}

\section{Minimum-Latency FEC Design with Delayed Feedback}\label{sec:DelayedFB}
In this section, we consider the minimum-latency code design in the presence of delayed feedback, which specifies whether a packet is successfully received at the receiver or not. Note that when the feedback is instantaneous, i.e., the receiving status of the packet sent during the $i$th time slot is available at the transmitter for making the decision at the $(i+1)$th time slot, the state in $\mathrm{MDP}_{1}$ is known exactly and hence we can implement the optimal actions following ARQ, which is given by
\begin{align}
\pi^*_{MDP}(s=(n,w,d))=\begin{cases}
0, & \textnormal{if $w=0$}\\
1, & \textnormal{otherwise}.
\end{cases}
\end{align}
When the feedback is delayed due to, e.g., the propagation or processing delays, the transmitter cannot tell the state in $\mathrm{MDP}_1$ exactly since the receiving status of some packets remains unknown. However, the delayed feedback information can still be used to reduce the uncertainty in the belief, and hence help the transmitter make better decisions. In the section, we assume that the receiving status of the $i$th transmission (i.e., the packet sent during time slot $i$) is available at the transmitter via the feedback channel at the beginning of the $(i+T+1)$th time slot, where $T$ is referred to as the ``feedback delay". $T=0$ corresponds to the case of instantaneous feedback.


\subsection{Mathematical Modeling}\label{sec:math_FB}
In the delayed feedback case, at the beginning of the $t$th time slot, the transmitter can observe the receiving status of the packet sent during the $(t-T-1)$th time slot via the feedback channel, where $T$ is the feedback delay. Hence, we can model the low latency code design with delayed feedback by a POMDP, which includes $\mathrm{MDP}_1$ and
\begin{itemize}
\item{The set of observation $\mathcal{Z}=\{0,1\}$, where $z_{t-T-1}^{t}\in\mathcal Z$ is the receiving status of the packet sent during the $(t-T-1)$th time slot and observed by the transmitter at the beginning of the $t$th time slot.  }
\item{The observation function $O(z)$, which is independent of the current state and is given by
\begin{align}
O(z)=\begin{cases}
1-p, & \textnormal{if $z=1$}\\
p, & \textnormal{if $z=0$}.
\end{cases}
\end{align}
}
\end{itemize}

In other words, at the beginning of the $t$th time slots, the transmitter has the knowledge of the receiving status of all the packets sent during the time slot $\{1,...,t-T-1\}$. Hence, we know the exact state of $\mathrm{MDP}_1$ at the beginning of the $(t-T)$th time slot, which is denoted as $\hat{s}_{t-T}$ for convenience. Then, we can update the belief state at the $t$th time slot based on $\hat{s}_{t-T}$ and the set of subsequent actions $\{a_i, i=t-T,...,t-1\}$ from the state transition diagram shown in Fig.~\ref{F:transition}. For notational convenience, we denote by $\hat{b}_i^{t}$ our estimation of belief state at time $i$, based on the information available at time $t$. Then we have
\begin{align}
\hat{b}_{t-T}^t(s)=
\begin{cases}
1, & \textnormal{if $s=\hat{s}_{t-T}\in\mathcal S$}\\
0, & \textnormal{otherwise}.
\end{cases}
\end{align}
Moreover, we can also estimate the belief at time $\{t-T+i, i=1,...,T\}$ via the following recursive formula
\begin{align}
b_{t-T+i}^t(s)=\sum_{s'\in\mathcal S}T(s,a_{t-T+i-1},s')\hat{b}_{t-T+i-1}^t(s'), \label{eq:updateBelief_FB}
\end{align}
where $a_{t-T+i-1}\in\mathcal A$ is the action applied during $(t-T+i-1)$th time slot and $T(s,a,s')$ is the state transition function for $\mathrm{MDP}_1$ given in \eqref{eq:transitionFun}.

In the absence of feedback, our belief state $b_t$ is obtained from the initial belief, i.e., $\mathbf{b}_1^{0}(0,0)=1$, and all the actions from time 1 until time $t-1$. With delayed feedback, at time $t$, the transmitter can make the decision based on the belief $\hat{b}_t^t$ from \eqref{eq:updateBelief_FB}, which has much less uncertainty as compared with $b_t$ in the absence of feedback,


\vspace{-0.2in}
\subsection{Efficient Algorithms}\label{sec:alg_FB}
The majority vote policy defined in \eqref{eq:MVPolicy} can be directly applied to the case with delayed feedback by using the more accurate belief state $\hat{b}_t^t$ at time $t$ as input in \eqref{eq:MVPolicy}. In this section, we focus on the extension of the $D$-step search algorithm to the case with delayed feedback. At each time slot, a belief tree by looking $D$ steps ahead will be built to assist the decision making process. According to the availability of feedback information, we have the following three cases:
\begin{itemize}
\item{For time slot $t\in\{1,2,...,T-D+1\}$, there is no feedback available even after looking at $D$ steps ahead. Hence, the belief tree is exactly the same as the case without feedback, which is discussed in Section~\ref{sec:alg_noFB}. Each action node is followed by one belief node.  }
\item{For time slot $t\in\{T-D+2,...,T+1\}$, there is no observation at the beginning of the $t$th time slot. Hence, $\hat{b}_t^t$ can be calculated based on the historical actions only. The first part of the tree follows from the case without feedback, where each belief node is expanded to two belief nodes. However, some feedback information will be received in the future steps that we look ahead. Specifically, there are two possible observations at the beginning of the $(T+2)$th time slot and all the subsequent time slots, i.e., $z=0$ or $z=1$. For each observation, we can estimate the corresponding updated belief. Then, for each belief, we have two possible actions, i.e., $a\in\{0,1\}$, and each followed by two possible observations again. Hence, for the second part of the tree, each belief node is expanded into 4 belief nodes in the next step.}
\item{For time slot $t\in\{T+2,...\}$, we calculate $\hat{b}_t^t$ based on the observation according to \eqref{eq:updateBelief_FB}. At each belief node, there are two action nodes, each followed by two observations. Hence, each belief node is expanded to 4 belief nodes in the next step. }
\end{itemize}

Regardless how the belief is estimated, the upper bound of the belief value can be calculated in the same way as in the scenario without feedback. Specifically, for the fringe belief nodes, the upper bound of the belief value is calculated by assuming that all the future states are directly observable, i.e., according to \eqref{eq:ValueUB}. On the other hand, the lower bound of the belief value needs to be estimated with the majority vote policy. However, for different possible observations in the future states, the majority vote policy may lead to different action sequences, and each action sequence corresponds to an approximated value of the current belief. The lower bound of the belief value can be obtained by averaging over all possible action sequences. In general, the number of action sequences that need to be considered grows exponentially with the action steps until reaching the final state, which leads to prohibitive complexity for large $N$. Therefore, we propose to estimate the lower bound of the belief value by ignoring the observations in the future states. Specifically, for belief state $\hat{b}_t^t$, one unique action sequence is determined based on the majority vote policy in \eqref{eq:MVPolicy} by updating the belief state as if there will be no observations in the future steps. After calculating the lower and upper bound of the belief value for the fringe nodes in the belief tree, the lower and upper bounds of the belief values for the parent nodes can be evaluated according to
\begin{align}
&L(b)=\begin{cases}
v^{L}(b), & \textnormal{if $b$ is a fringe node}\\
\max_{a\in\mathcal{A}}L(b,a), & \textnormal{otherwise},
\end{cases}\label{eq:lowerBound_FB}\\
&\textnormal{with }L(b,a)\triangleq\sum_{s\in\mathcal S}b(s)R(s,a)+\sum_{z\in\mathcal Z}O(z)L(\tau(b,a,z));\\
&U(b)=\begin{cases}
v^{U}(b), & \textnormal{if $b$ is a fringe node}\\
\max_{a\in\mathcal{A}}U(b,a), & \textnormal{otherwise},
\end{cases}\label{eq:upperBound_FB}\\
&\textnormal{with }U(b,a)\triangleq\sum_{s\in\mathcal S}b(s)R(s,a)+\sum_{z\in\mathcal Z}O(z)U(\tau(b,a,z)),\label{eq:beliefActionValue_FB}
\end{align}
where $R(b,a)$ is the immediate reward, $O(z)$ is the probability for observing $z$, and $\tau(b,a,z)$ is the general belief update function with observation. For the problem considered in this section, the updated belief $\tau(b,a,z)$ is calculated according to \eqref{eq:updateBelief_FB}. Then, the Branch-and-Bound pruning algorithm can be applied to reduce the size of the belief tree and  the action is chosen based on the belief-action values of the current state, according to \eqref{eq:DStepPolicy}.

\begin{example}
Consider an illustrative example where $N=3$ packets need to be sent via the channel with packet erasure probability $p$. The feedback delay is $T=2$ time slots. In the first time slot, an information packet is sent. At the beginning of the 2nd time slot, no feedback information has arrived yet and  the belief state is estimated to be $b_2=\mathbf{b}_2^1=[1-p,p]$. Then, if action $a=0$ is applied, we have belief $b_3=\mathbf{b}_3^2=[(1-p)^2,p(1-p),p(1-p),p^2]$. If action $a=1$ is applied, we have belief $b_3=\mathbf{b}_3^1=[1-p^2,p^2]$. With $T=2$, we expect to observe the receiving status of the first packet at the beginning of the 4th time slot. Hence, with the $2$-step search algorithm, we can build the belief tree as shown in Fig.~\ref{F:tree_FB}. Take the fringe nodes $\mathbf{b}_4^2$ and $\tilde{\mathbf{b}}_4^2$ as an example. They are estimated based on the same sequence of actions $\{a_2=0,a_3=1\}$ and different observations. If we observe $z_1^4=1$, the estimated belief state is $\mathbf{b}_4^2=[1-p^2, p^2, 0, 0]$. If we observe $z_1^4=0$, the estimated belief state is $\tilde{\mathbf{b}}_4^2=[(1-p)^2,0,2p(1-p),p^2]$. With $p=0.3$, the values (value ranges) of the belief states are labeled in Fig.~\ref{F:tree_FB}. The lower and upper bounds are propagated from the fringe belief nodes to the parent nodes according to \eqref{eq:lowerBound_FB}-\eqref{eq:beliefActionValue_FB}. Then, based on \eqref{eq:DStepPolicy}, the action obtained by the 2-step search algorithm at $t=2$ is $a=0$, i.e., a new information packet should be sent during the 2nd time slot.

\begin{figure}[htb]
\centering
\includegraphics[scale=0.5]{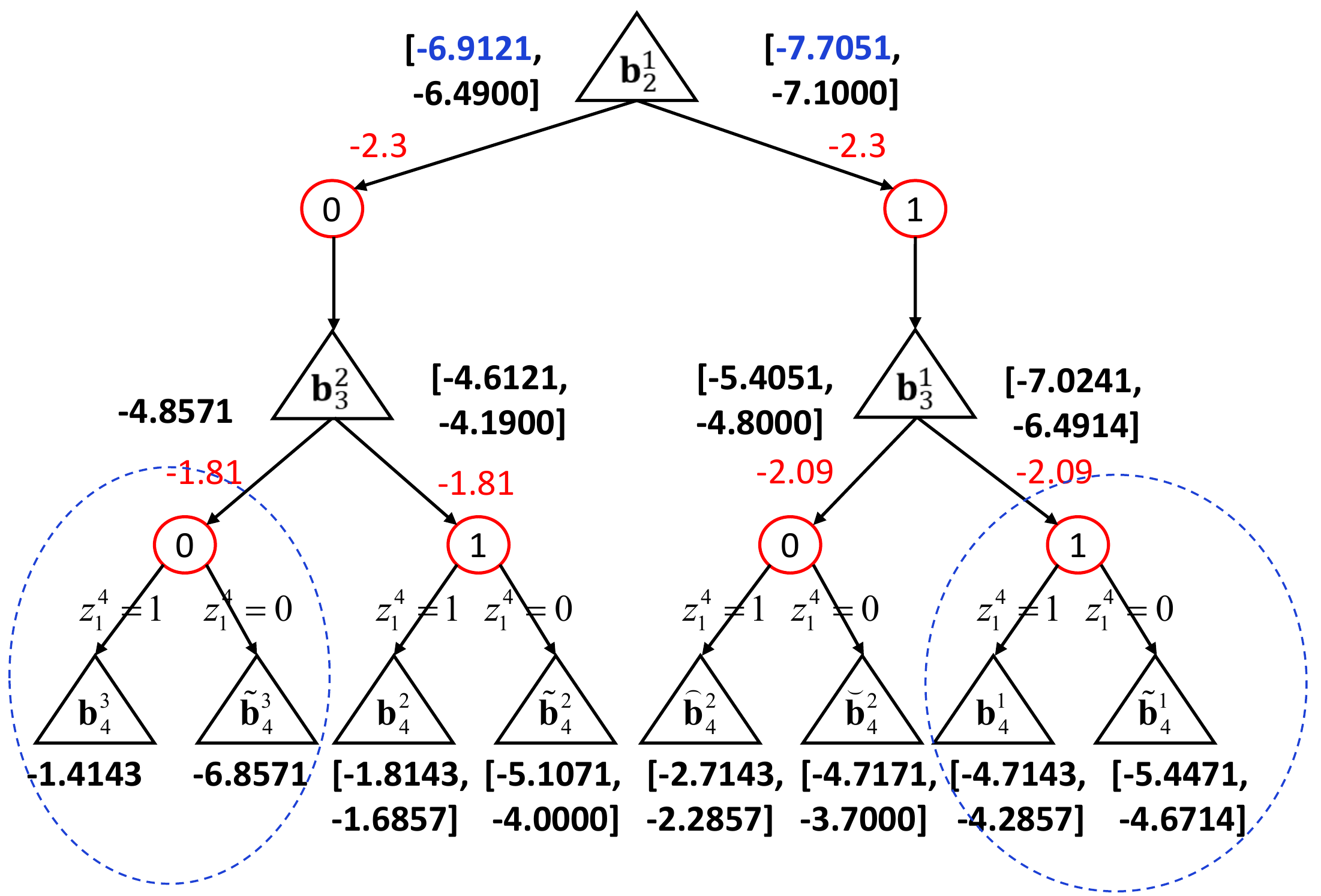}
\caption{The 2-step belief tree for $N=3$ at $t=2$ with feedback delay $T=2$. The triangles denote the belief states and the circles denote the actions. For $p=0.3$, the expected immediate reward for the belief-action pairs are labeled with red text, and the (range of the) optimal belief state values are labeled with bold black text.  }
\label{F:tree_FB}
\end{figure}
\end{example}

\section{Numerical Examples}\label{sec:simu}
In this section, numerical results are provided to evaluate the proposed low latency FEC design by solving the corresponding POMDP. First, we consider the case without feedback. The proposed algorithms discussed in Section~\ref{sec:NoFB} are compared with the benchmark scheme introduced in \cite{LowDelayFEC}, which is referred to as ``Low-Delay FEC". With Low-Delay FEC, a coded packet is sent after every $L-1$ information packets are sent, where $L$ is a positive integer with its value chosen according to the channel statistics \cite{LowDelayFEC}. First, we compare the coding decisions obtained by using different strategies in Fig~\ref{F:CodingView}. For each coding strategy, after all the information packets have been transmitted, i.e, when the state is $(N,\cdot,\cdot)$, the transmitter will continue sending the coded packets until the receiver is able to decode all the information packets\footnote{In the absence of feedback, the transmitter can estimate the probability for reaching the final state $(N,0,0)$ and terminate the transmission once the decoding probability is sufficiently large.}. It is observed from Fig~\ref{F:CodingView} that, before reaching the state $(N,\cdot,\cdot)$, the codes obtained by solving the POMDP have adaptive coding rate, with more information packets sent at the beginning of the transmission process whereas more coded packets sent at the later stage. This is expected since there are more packets queueing at the transmitter side at the beginning of the process, whereas more packets are waiting at the receiver side to be delivered to the application layer in the later transmission stage.
\begin{figure}[htb]
\centering
\includegraphics[scale=0.7]{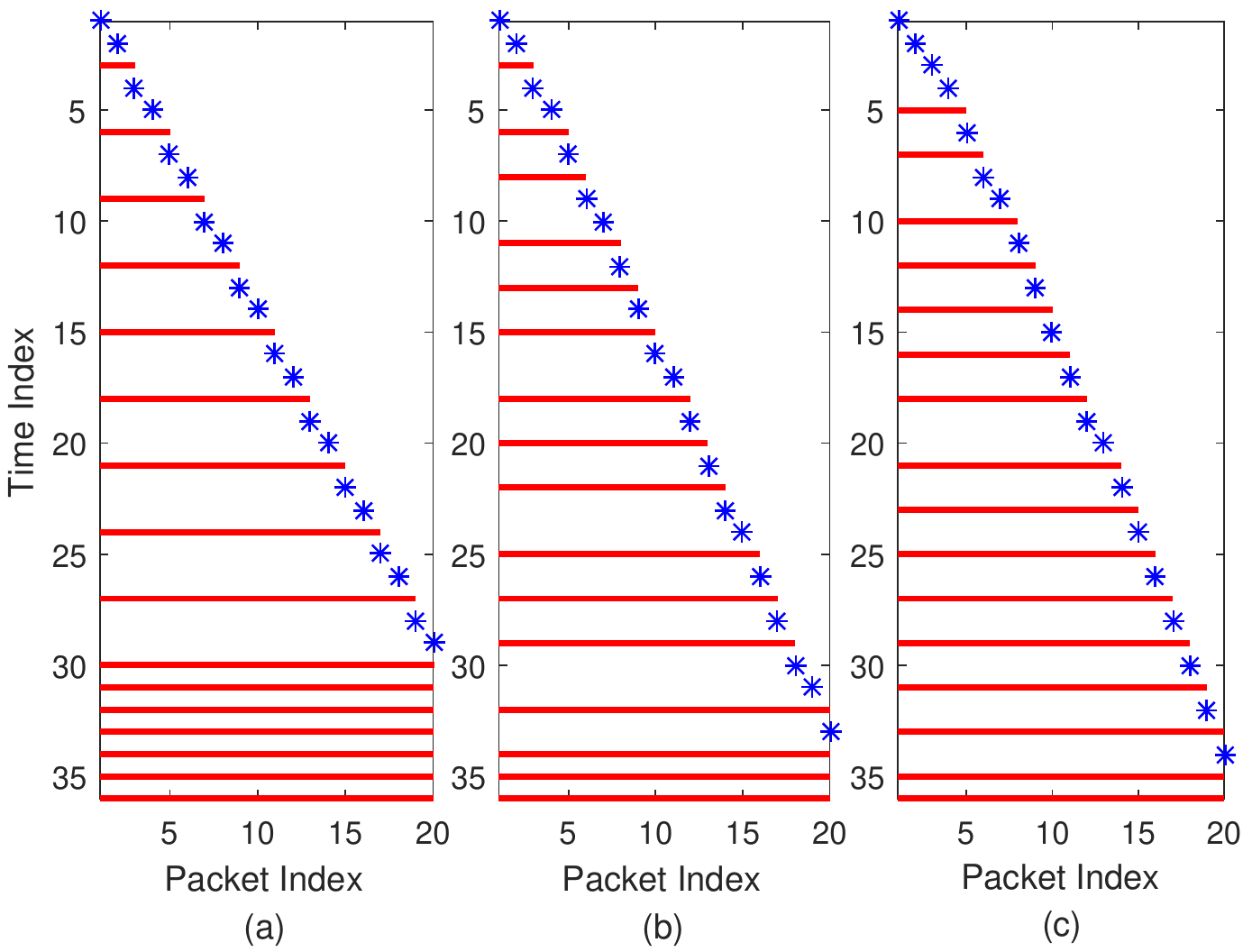}
\caption{The low latency FEC for $N=20$ and $p=0.3$ obtained from different strategies: (a) Low-Delay FEC with $L$=3; (b) Majority Vote; (c) 2-step Search. The blue dots denote the information packet and red lines represent the coded packets with the corresponding spans.  }
\label{F:CodingView}
\end{figure}

In Fig.~\ref{F:Results_noFB}(a), we compare the average end-to-end delay achieved by the proposed schemes with the optimal coding policy obtained via brute force search, for sending $N=15$ packets. It is observed that the proposed code designs have very close performance to that obtained via brute force search, while taking much less computations. In Fig.~\ref{F:Results_noFB}(b), we compare the average end-to-end delay achieved by the proposed schemes with Low-Delay FEC \cite{LowDelayFEC} for sending $N=100$ packets, where the brute force search is not feasible. It is observed that the performance of the Low-Delay FEC is critically dependent on the choice of $L$. By solving the POMDP with simple majority vote policy, the resulting code strictly outperforms the Low-Delay FEC with the optimal choice of $L$ at any given erasure probability. We also observe that the $2$-step algorithm can lead to better code design than the majority vote policy, at the cost of increased complexity. If we further increase the search steps from 2 to 4, the performance improvement  is negligible as shown in Fig.~\ref{F:Results_noFB}. It was proved in \cite{Ross2008} that the $D$-step search algorithm is asymptotically optimal, i.e., as the number of the search steps $D$ increases, the solution approaches the optimal policy for the POMDP. Therefore, we conjecture that the code design with the $2$-step search algorithm is close to optimal.

\begin{figure}[h]
\centering
\subfloat[][]{%
\adjustbox{valign=b}{
\includegraphics[width=0.45\textwidth]{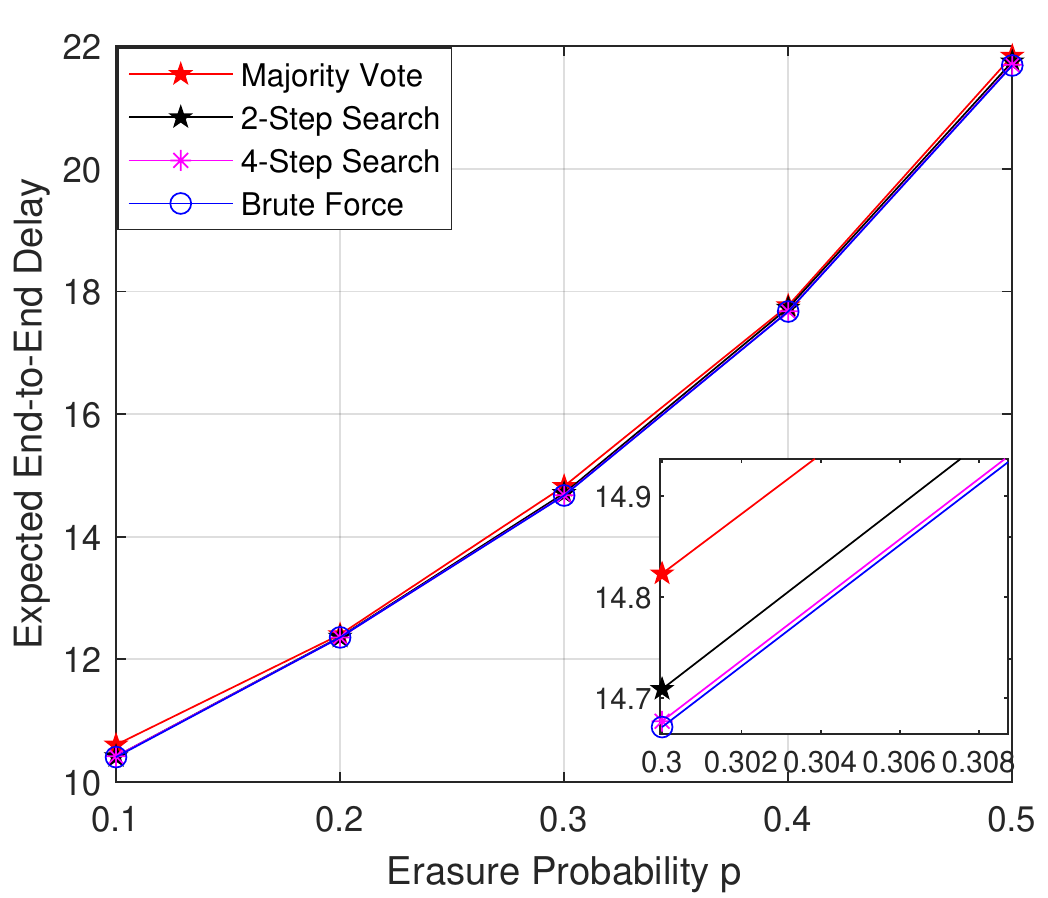}}}
\hspace{0.03cm}
\subfloat[][]{%
\adjustbox{valign=b}{
\includegraphics[width=0.45\textwidth]{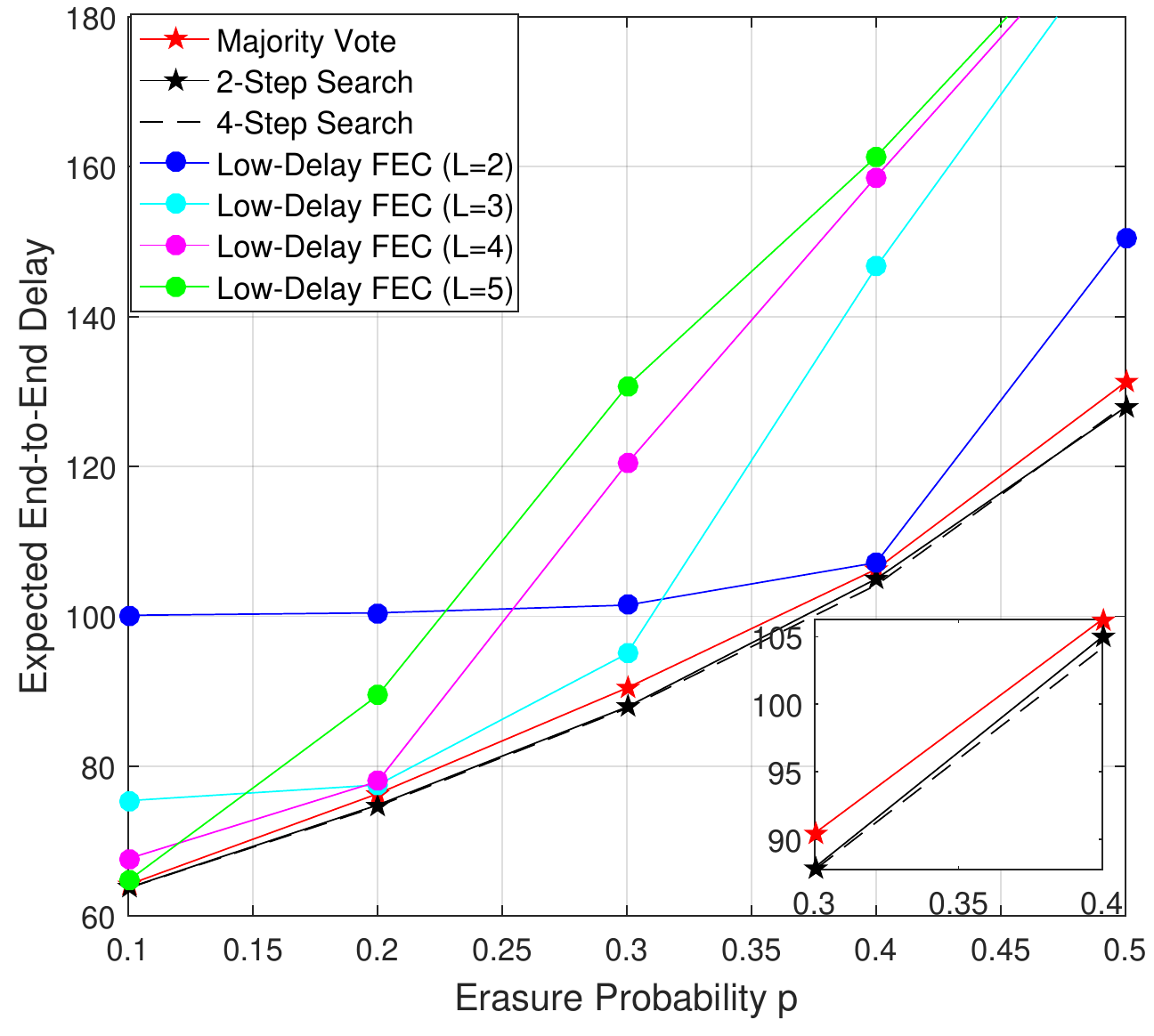}}}
\caption{End-to-end delay achieved by various code designs without feedback. (a) $N=15$. (b) $N=100$. }
\label{F:Results_noFB}
\end{figure}

Next, we consider the minimum-latency FEC design with delayed feedback. Note that the Low-Delay FEC proposed in \cite{LowDelayFEC} has a fixed code structure, which cannot utilize the feedback information. Hence, the performance of the Low-Delay FEC cannot be further improved even in the presence of feedback. In contrast, our proposed approach directly utilizes the feedback information, as discussed in Section~\ref{sec:DelayedFB}. To evaluate our proposed code design with delayed feedback, we compare it with the existing scheme proposed in \cite{Garrido2018}, which is referred to as ``Feedback-based adaptive coding"\footnote{The coding scheme proposed in \cite{Garrido2018} works for general packet arrival models. For comparison, we apply the coding scheme proposed in \cite{Garrido2018} with the packet arrival model where $N$ packets arrive at $t=0$ and no packet arrives for $t>0$. }. In \cite{Garrido2018}, the author proposes to estimate the receiver queue length at the transmitter side based on the delayed feedback information and the channel statistics. Then, a coded packet is sent if the estimated receiver queue length is above a certain threshold. The optimal threshold was shown to be 1  for achieving the minimum end-to-end latency in \cite{Garrido2018}. In Fig.~\ref{F:Results_FB}, we compare the average end-to-end latency achieved by the proposed coding schemes and the Feedback-based adaptive coding for sending $N=100$ packets with feedback delay $T=2$. The simulation results are obtained by averaging over 1000 realizations. It is observed that when the packet erasure probability is low, the FEC designed with the majority vote policy has a similar performance as the Feedback-based adaptive coding proposed in \cite{Garrido2018}. However, when the erasure probability is large, the majority vote policy significantly outperforms the Feedback-based adaptive coding. Furthermore, the $2$-step search algorithm has a constant performance gain over the majority vote policy for a wide range of erasure probabilities.

\begin{figure}[htb]
\centering
\includegraphics[scale=0.9]{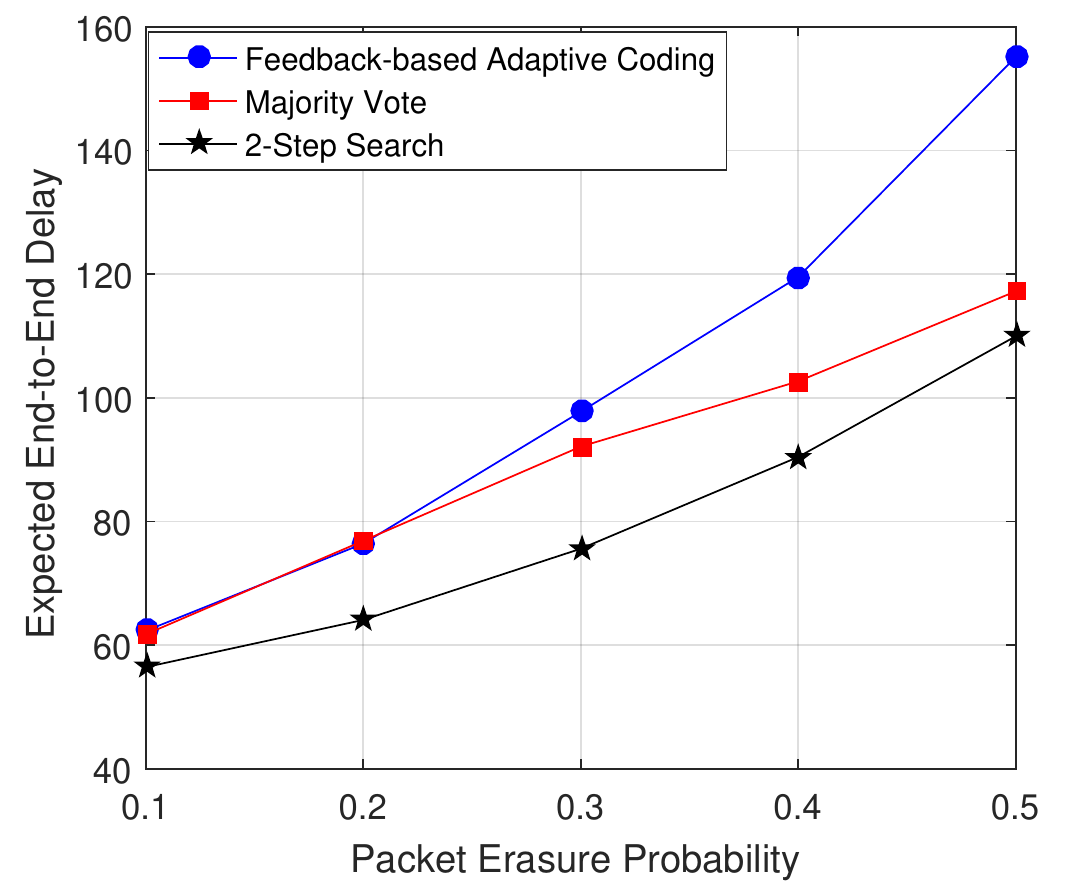}
\caption{Average end-to-end delay achieved by various code designs for sending $N=100$ packets with feedback delay $T=2$. }
\label{F:Results_FB}
\end{figure}

In Fig.~\ref{F:VarFB}, we plot the average end-to-end latency achieved by the $2$-step search algorithm for sending $N=100$ packets with different feedback delays. Note that $T=0$ corresponds to the case with instantaneous feedback, for which the $2$-step search algorithm reduces to ARQ. On the other hand, $T=\infty$ corresponds to the case without feedback. As expected, the latency performance degrades as the feedback delay increases. Furthermore, a substantial performance degradation is observed when the feedback delay increases from $T=0$ to $T=2$. As $T$ further increases, the performance degradation becomes smaller. For the delayed feedback with feedback delay $T=8$, we can still observe a significant performance gain as compared with the case without any feedback.

\begin{figure}[htb]
\centering
\includegraphics[scale=0.8]{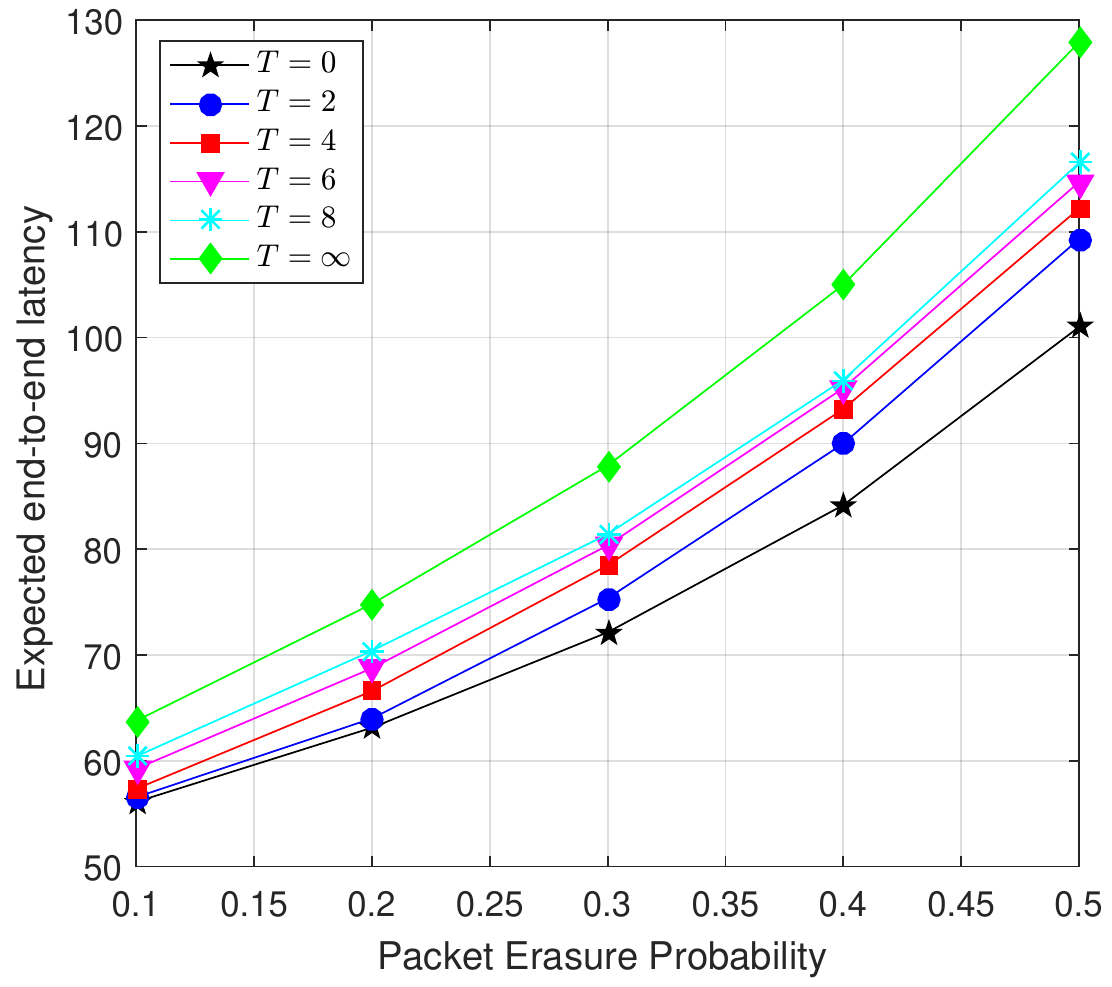}
\caption{Average end-to-end latency achieved by 2-step search algorithm for various feedback delays, with $N=100$. }
\label{F:VarFB}
\end{figure}

\section{Conclusion}\label{sec:con}
In this paper, we have investigated the optimal FEC design for achieving the minimum average end-to-end latency over a lossy channel without or with delayed feedback. We formulated the FEC design problem as a POMDP  and proposed the solutions  with the majority vote policy and the $D$-step search algorithm. The majority vote policy can be implemented with a very low complexity, whereas the $D$-step algorithm can achieve the close-to-optimal solution with a higher complexity. The codes constructed by solving the POMDP were shown to outperform the existing codes for both scenarios without and with delayed feedback. Although we have focused on the i.i.d. packet erasures in this paper, the proposed POMDP formulation and the solutions can be extended to the case with bursty erasures, e.g, the Gilbert-Elliott channel, by including the channel states into the environment states of the POMDP.

In the future, we will extend the proposed solutions for  adaptive packet arrival, in which case the number of states may go to infinity, and investigate the tradeoff between the achievable rate and latency in such scenario.  Furthermore, we will also consider the optimal FEC design for achieving the minimum \emph{guaranteed latency}, which is a function of the expectation and variance of the end-to-end latency.

\appendices
\section{Proof of Lemma~\ref{lem:finalValue}}\label{A:finalValue}
First, it is noted the state $s=(N,0,0)$ is the final state with no further transitions and hence we have $v^*(N,0,0)=0$. Then, according to the transition diagram shown in Fig.~\ref{F:transition}(c), after applying the optimal action $a=1$, we have
\begin{align}
v^*(N,w,w-1)=-w+(1-p)v^*(N,0,0)+pv^*(N,w,w-1),
\end{align}
which renders $v^*(N,w,w-1)=-\frac{w}{1-p}$. Then, according to the transition diagram shown in Fig.~\ref{F:transition}(b), we have
\begin{align}
v^*(N,w,d)=-w+(1-p)v^*(N,w,d+1)+pv^*(N,w,d), \forall d<w,
\end{align}
which leads to
\begin{align}
v^*(N,w,d)&=-\frac{w}{1-p}+v^*(N,w,d+1)\nonumber\\
&=-\frac{w(w-d-1)}{1-p}+v^*(N,w,w-1)=-\frac{w(w-d)}{1-p}, \forall d<w-1.
\end{align}
It can be easily verified that $v^*(N,w,d)=-\frac{w(w-d)}{1-p}$ holds for all the valid states $(N,w,d)$ with $0\leq d<w\leq N$.

\section{Proof of Lemma~\ref{lem:OptValues}}\label{A:OptValue}
If $\mathrm{MDP}_1$ is fully observable, the optimal policy can be easily determined, which is sending the information packet if and only if $w=0$ and sending the coded packets otherwise. Hence, according to the transition diagram shown in Fig.~\ref{F:transition}, we have
\begin{align}
&v^*_{MDP}(n,0,0)=-(N-n)+(1-p)v^*_{MDP}(n+1,0,0)+pv^*_{MDP}(n,0,0),\label{eq:state00}\\
&v^*_{MDP}(n,w,w-1)=-(N-n+w)+(1-p)v^*_{MDP}(n,0,0)+pv^*_{MDP}(n,w,w-1),\label{eq:statewm1}\\
&v^*_{MDP}(n,w,d)=-(N-n+w)+(1-p)v^*_{MDP}(n,w,d+1)+pv^*_{MDP}(n,w,d),\forall d<w-1,\label{eq:statewd}.
\end{align}
From \eqref{eq:state00}, we have
\begin{align}
v^*_{MDP}(n,0,0)&=-\frac{N-n}{1-p}+v^*_{MDP}(n+1,0,0)\nonumber\\
&=-\sum_{i=n}^{N-1}\frac{N-i}{1-p}+v^*_{MDP}(N,0,0)\nonumber\\
&=-\frac{(N-n+1)(N-n)}{2(1-p)}\label{eq:valueMid00}
\end{align}
where the last equality follows from algebraic summation and the fact that the state $(N,0,0)$ is the final state with the optimal value $v^*_{MDP}(N,0,0)=v^*(N,0,0)=0$.
Then, by substituting \eqref{eq:valueMid00} into \eqref{eq:statewm1}, we have
\begin{align}
v^*_{MDP}(n,w,w-1)=-\frac{N-n+w}{1-p}-\frac{(N-n+1)(N-n)}{2(1-p)}\label{eq:valueMidwm1}.
\end{align}
Then, with \eqref{eq:valueMidwm1} and the recursive formula \eqref{eq:statewd}, we have
\begin{align}
v^*_{MDP}(n,w,d)&=-\frac{N-n+w}{1-p}+v^*_{MDP}(n,w,d+1)\nonumber\\
&=-\frac{(N-n+w)(w-d-1)}{1-p}+v^*_{MDP}(n,w,w-1)\nonumber\\
&=-\frac{(N-n+w)(w-d)}{1-p}-\frac{(N-n+1)(N-n)}{2(1-p)}, \forall d<w-1. \label{eq:ValueWD}
\end{align}
It can be verified from \eqref{eq:valueMid00} and \eqref{eq:valueMidwm1} that the expression \eqref{eq:ValueWD} holds for all the valid states $(n,w,d)$ with $0\leq n\leq N$,$0\leq w\leq N$ and $0\leq d<w$.

\bibliographystyle{ieeetr}
\bibliography{MachineLearning}

\end{document}